%% file: recoverable.tex
\documentclass[sigconf,nonacm=true]{acmart} 

\usepackage{amsthm}
\usepackage{amsmath}
\usepackage{graphicx}
\usepackage{enumerate}

\newcommand{\arxcam}[2]{#1}  

\copyrightyear{2022}
\acmYear{2022}
\setcopyright{acmlicensed}\acmConference[PODC '22]{Proceedings of the 2022 ACM Symposium on Principles of Distributed Computing}{July 25--29, 2022}{Salerno, Italy}
\acmBooktitle{Proceedings of the 2022 ACM Symposium on Principles of Distributed Computing (PODC '22), July 25--29, 2022, Salerno, Italy}
\acmPrice{15.00}
\acmDOI{10.1145/3519270.3538418}
\acmISBN{978-1-4503-9262-4/22/07}
\input{code3}

\newtheorem{theorem}{Theorem}
\newtheorem{definition}[theorem]{Definition}
\newtheorem{lemma}[theorem]{Lemma}
\newtheorem{corollary}[theorem]{Corollary}
\newtheorem{observation}[theorem]{Observation}
\newtheorem{proposition}[theorem]{Proposition}

\title{When Is Recoverable Consensus  Harder Than Consensus?}
\author{Carole Delporte-Gallet}
\orcid{0000-0001-7946-6708}
\affiliation{\institution{Universit\'{e} Paris Cit\'{e}}
	\department{IRIF}
	\city{Paris}
	\country{France}}
\author{Panagiota Fatourou}
\orcid{0000-0002-6265-6895}
\affiliation{\institution{Universit\'{e} Paris Cit\'{e}}
	\department{LIPADE}
	\city{Paris}
	\country{France}}
\affiliation{\institution{FORTH ICS}
	\city{}
	\country{}
}
\affiliation{\institution{University of Crete}
	\city{Heraklion}
	\country{Greece}}
\author{Hugues Fauconnier}
\orcid{0000-0002-3250-716X}
\affiliation{\institution{Universit\'{e} Paris Cit\'{e}}
	\department{IRIF}
	\city{Paris}
	\country{France}}
\author{Eric Ruppert}
\orcid{0000-0001-5613-8701}
\affiliation{\institution{York University}
	\city{Toronto}
	\country{Canada}}

\newcommand{\op}[1]{{\sc #1}}
\newcommand{\ignore}[1]{}

\newcommand{\here}[1]{}

\newcommand{\floor}[1]{\left\lfloor #1 \right\rfloor}
\newcommand{\ceil}[1]{\left\lceil #1 \right\rceil}

\newcommand{\y}[1]{{\color{blue} #1}\normalcolor}
\newcommand{\new}[1]{{\color{red} #1}\normalcolor}
\newcommand{\x}[1]{{\color{purple} #1}\normalcolor}   
\renewcommand{\y}[1]{#1}
\renewcommand{\new}[1]{#1}
\renewcommand{\x}[1]{#1}

\newcommand{\run}{run}
\newcommand{\cond}[1]{{#1}-recording}

\newcommand{\Perform}{{\tt Perform}}
\newcommand{\RUniversal}{{\sf RUniversal}}



\begin{document}
\fancyhead{}

\begin{abstract}
We study the ability of different shared object types to solve recoverable consensus
using non-volatile shared memory in a system with crashes and recoveries.
In particular, we compare the difficulty of solving recoverable consensus
to the difficulty of solving the standard wait-free consensus problem in a system
with halting failures.
We focus on the model where individual processes may crash and recover
and the large
class of object types that are equipped with a read operation.
We characterize the readable object types that can
solve recoverable consensus among a given number of processes.
Using this characterization, we show that the number of processes that can solve consensus 
using a readable type can be larger than the number of processes that 
can solve recoverable consensus using that type, but only slightly larger.
\end{abstract}

\maketitle

\here{Must update location of full version \cite{full}}

\input{intro}
\input{hierarchy}

\input{readable-intro}

\input{readable-sufficient}
\input{readable-necessary}

\input{readable-consequences}

\input{universality}


\input{conclusion}

\begin{acks}
\x{
This research was conducted while Eric Ruppert was visiting the Universit\'{e} Paris Cit\'{e}, with funding from 
\grantsponsor{bad}{IDEX-Universit\'{e} Paris Cit\'{e}}{https://u-paris.fr/}
project 
\grantnum{bad}{BAD} and
\grantsponsor{ducat}{ANR DUCAT}{https://www.irif.fr/anr-ducat/}
project number
\grantnum{ducat}{-20-EC48-0006}.
Support was also provided by the
\grantsponsor{nserc}{Natural Sciences and Engineering Research Council}{https://www.nserc-crsng.gc.ca/} 
of Canada, the
\grantsponsor{curie}{Marie Sklodowska-Curie}{https://marie-sklodowska-curie-actions.ec.europa.eu/} project
\grantnum{curie}{PLATON (GA No 101031688)},
and 
\grantsponsor{hfri}{HFRI}{https://www.elidek.gr/en/}
under the 2nd Call for HFRI Research Projects to support faculty members and researchers (project number
\grantnum{hfri}{3684}).
}
\end{acks}

\pagebreak[1]

\bibliographystyle{plain}
\bibliography{recoverable}

\arxcam{\input{appendix}}{}

\end{document}

%% file: code3.tex
\newcounter{linenum}
\def\codeTabSpace{\hspace*{6mm}}
\newenvironment{code}%
{\begin{tabbing}%
\codeTabSpace \= \hspace*{30mm} \= \hspace*{40mm} \= \hspace*{42mm} \= \kill%
}%
{\end{tabbing}%
}
\newcounter{ind}
\newcommand{\n}{\addtocounter{ind}{7}\hspace*{7mm}}
\newcommand{\p}{\addtocounter{ind}{-7}\hspace*{-7mm}}
\newcommand{\nl}{\\\stepcounter{linenum}{\scriptsize\arabic{linenum}}\>\hspace*{\value{ind}mm}}

\newcommand{\bl}{\\[-1.5mm] \>\hspace*{\value{ind}mm}}
\newcommand{\firstline}{\stepcounter{linenum}{\scriptsize \arabic{linenum}}\>}
\newcommand{\llabel}[1]{\addtocounter{linenum}{-1}\refstepcounter{linenum}\label{#1}}
\newcommand{\lref}[1]{\ref{#1}} 

%% file: intro.tex

\section{Introduction}

Recoverable consensus can play a key role in the study of asynchronous systems with non-volatile
shared memory where processes can crash and recover, just as the standard consensus problem plays a 
central role in the study of asynchronous systems where processes may halt.
In this paper, our goal is to leverage extensive research on the 
solvability of the standard consensus problem in systems equipped with different types of shared
objects to gain knowledge about recoverable consensus in systems with non-volatile  memory.

We consider an asynchronous model of computation, where processes communicate with one another
by accessing shared memory.  
In particular, we are interested in \x{studying how concurrent} algorithms can take advantage
of recent advances in non-volatile main memory, which maintains its stored values 
even when its power supply is turned off.
This allows for algorithms that can carry on with a computation when
processes crash and recover.
We consider \y{a standard theoretical model~\cite{ABH18,Gol20,GH17,GH18} for this setting},
where each process's local memory is volatile, but shared memory is non-volatile,
{\y and processes may crash and recover individually in an asynchronous manner.
After} a process crashes, its local memory, including its programme counter,
is reinitialized to its initial state when the process recovers.
Process crashes do not affect the state of shared memory.
At recovery time, the process begins to execute its code again from the beginning\footnote{\y{Alternatively, 
it could execute a recovery function. Our results hold
either way. We use the simpler assumption of re-starting \x{upon} recovery
to prove our  results. 
}}.
We refer to the sequence of steps that a process takes between crashes as a \emph{\run} of its code.

The consensus problem, where each process gets an input and all processes must agree to output one of them,
has been central to the study of shared-memory computation 
in asynchronous systems with process halting failures (but no recoveries).
A shared object type is defined by a sequential specification, which specifies
the set of possible states of the object, the operations that can be performed on it,
and how the object changes state and returns a response when an operation is applied on it.
Herlihy \cite{Her91} defined the consensus number of a type $T$, denoted $cons(T)$,
to be the maximum number of processes that can solve consensus using objects of type $T$ and
read/write registers, or $\infty$ if there is no such maximum.
The classification of  types according to their consensus number is called 
the \emph{consensus hierarchy}.
This classification is particularly meaningful because of Herlihy's universality result:  a type $T$ can
be used (with registers) 
to obtain wait-free implementations of {\it all} object types in a system of $n$ processes
if and only if $cons(T)$ is at least~$n$.

\begin{figure*}[t]
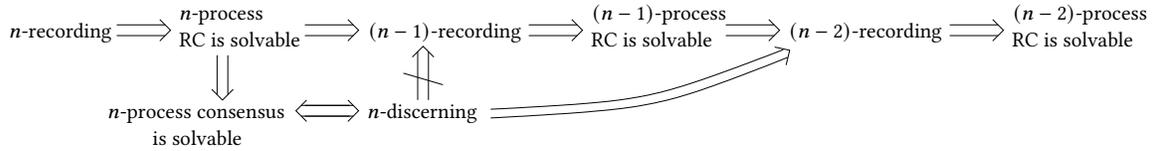
\Description{Figure showing relationships:  n-recording implies solvability of n-process RC, which implies (n-1)-recording.  Solvability of n-process RC also implies solvability of n-process consensus, which is equivalent to n-discerning.  n-discerning does not imply (n-1)-recording but does imply (n-2)-recording.}
\caption{Relationships between conditions and solvability of consensus and recoverable consensus (with independent crashes) using a deterministic, readable type.  \label{summary}}
\end{figure*}

Golab \cite{Gol20} defined the {\it recoverable consensus} (RC) problem, where processes must agree 
on one of their input values, even if processes may crash and recover.
An algorithm for  RC  defines a routine for each process to execute that takes an 
input value and eventually returns an output value, satisfying the following three properties.
\begin{itemize}
\item
Agreement:  no two output values produced  are different.  (This includes outputs by different processes and outputs of the same process when it performs multiple runs of the algorithm because it crashes and recovers.)
\item
Validity:  each output value is the input value of some process.
\item
Recoverable wait-freedom:  if a process executes its algorithm from the beginning, it either crashes or outputs a value after a finite number of its own steps.
\end{itemize}
Like Golab, we assume a process's input value does not change,
even across multiple {\run}s,
but this is not a crucial assumption.  
(If an RC algorithm requires this precondition,
it can be transformed into one that does not
using a register for each process's input.  When a process begins a \run, it reads this register
and, if it has not yet been written, the process writes its input value.  
It then uses
the value in the register as its input, ensuring that all of the process's
{\run}s of the original algorithm use the same input value.)
\new{Berryhill, Golab and Tripunitara \cite{BGT15} described how Herlihy's universality result
carries over to the model with crashes and recoveries, using RC in place
of consensus.
(See Section \ref{universal} for details.)}

There are two common failure models for crashes and recoveries: \y{\emph{simultaneous crashes}~\cite{IMS16}}, 
where all processes  crash simultaneously, and \y{\emph{independent crashes} (introduced in~\cite{GR16} to study recoverable mutual exclusion),}
where processes \y{can crash and recover individually in an asynchronous way}. 
Golab \cite{Gol20} 
defined two recoverable consensus hierarchies.
For an object type $T$, the \emph{simultaneous RC number} of $T$ is the maximum number
of processes that can solve RC using \x{an unbounded number of} shared objects of type $T$ and read/write registers
when simultaneous crashes may occur.
Similarly, the \emph{independent RC number} of $T$, which we denote $rcons(T)$, is the maximum number
of processes that can solve RC using shared objects of type $T$ and read/write registers
when independent crashes may occur.
In both cases, if no maximum exists we say the RC number is $\infty$.
This is a slight modification of Golab's 
definition.\footnote{Golab's definition of RC numbers required the RC algorithms to use a \emph{bounded} number of objects.  We permit an infinite number of objects. 
When Jayanti \cite{Jay97} formalized Herlihy's
consensus hierarchy, he similarly allowed an unbounded number of objects to be used in solving consensus.
(However, it follows from K\"onig's Lemma \cite{Kon27} that any wait-free algorithm for the standard 
consensus problem that uses objects with \x{finite} non-determinism will use finitely many objects.)
Universal constructions, which are one of the main motivations for studying the hierarchy,
require an infinite number of instances of consensus anyway,
so even if each instance
uses a \x{finite} number of objects, the overall construction would still use an infinite number.
}
\new{As an example, 
\arxcam{we show in Appendix \ref{stack-proof} that $rcons(stack)=1$,}{$rcons(stack)=1$ \cite{full},}
whereas it is known that 
$cons(stack)=2$~\cite{Her91}.}

\subsection{Our Results}

We focus on independent crashes since a simple extension \x{of} Golab's result \cite{Gol20}
described in Section \ref{simultaneous}
shows that RC  has exactly the same difficulty as consensus in a system with simultaneous crashes.

Our main results are for deterministic shared object types that are \emph{readable}, meaning
that they are equipped with a read operation that returns the current state of the object
without changing it.
We define, for all $n\geq 2$, the \cond{$n$} property for shared object types.
\x{Roughly speaking, a readable type $T$ is \cond{$n$} if $n$ processes can be divided into two teams
and use one object of type $T$ to determine which of the two teams ``wins'', even when processes crash and recover.
The first team to perform an update operation on the object is the winning team, and this 
information is \emph{recorded} in the object's state, so that processes can determine 
which team wins by reading the object.

We show in Section \ref{sec:sufficient} that being \cond{$n$} is sufficient for solving RC among $n$ processes.
We also show in Section \ref{sec:necessary} that the slightly weaker condition of being \cond{$(n-1)$} is necessary for solving RC
among $n$ processes.
Thus, we have a fairly simple way of determining the approximate value of $rcons(T)$:  if 
$T$ is \cond{$n$} but not \cond{$(n+1)$}, we know that $rcons(T)$ is either $n$ or $n+1$.}

\x{Our \cond{$n$} property is related to Ruppert's   $n$-discerning property \cite{Rup00},
which was defined to characterize readable types that can solve $n$-process consensus.
In Section \ref{sec:relationship}, we prove relationships between these two properties.}
This allows us to prove that if a type has consensus number $n$, then its RC
number is between $n-2$ and $n$.
We give examples of types $T$ with $rcons(T)=cons(T)$ and others with $rcons(T) < cons(T)$.
In Section \ref{sec:robustness}, we also use our characterization to show that weak types do not become much stronger
(in terms of their power to solve RC) when used together.
\new{Section \ref{universal} describes how Herlihy's motivation for studying the consensus hierarchy
carries over to the RC hierarchy for the setting of non-volatile memory.}
\x{See Figure \ref{summary} for an overview of our results.}

\ignore{
\here{Youla suggests we might want a table that summarizes our results. But I'm not sure how to design such a table.
PF: I like the description above and I do not think any more that a table is needed,
so I have removed this comment.}
}

%% file: summary.pdf_t
\begin{picture}(0,0)%
\includegraphics{summary.pdf}%
\end{picture}%
\setlength{\unitlength}{2960sp}%
\begingroup\makeatletter\ifx\SetFigFont\undefined%
\gdef\SetFigFont#1#2#3#4#5{%
  \reset@font\fontsize{#1}{#2pt}%
  \fontfamily{#3}\fontseries{#4}\fontshape{#5}%
  \selectfont}%
\fi\endgroup%
\begin{picture}(9353,1212)(4336,-5926)
\put(9226,-4861){\makebox(0,0)[lb]{\smash{{\SetFigFont{8}{9.6}{\familydefault}{\mddefault}{\updefault}{\color[rgb]{0,0,0}$(n-1)$-process}%
}}}}
\put(9226,-5086){\makebox(0,0)[lb]{\smash{{\SetFigFont{8}{9.6}{\familydefault}{\mddefault}{\updefault}{\color[rgb]{0,0,0}RC is solvable}%
}}}}
\put(10876,-5011){\makebox(0,0)[lb]{\smash{{\SetFigFont{8}{9.6}{\familydefault}{\mddefault}{\updefault}{\color[rgb]{0,0,0}\cond{$(n-2)$}}%
}}}}
\put(4351,-5011){\makebox(0,0)[lb]{\smash{{\SetFigFont{8}{9.6}{\familydefault}{\mddefault}{\updefault}{\color[rgb]{0,0,0}\cond{$n$}}%
}}}}
\put(5776,-5086){\makebox(0,0)[lb]{\smash{{\SetFigFont{8}{9.6}{\familydefault}{\mddefault}{\updefault}{\color[rgb]{0,0,0}RC is solvable}%
}}}}
\put(7351,-5011){\makebox(0,0)[lb]{\smash{{\SetFigFont{8}{9.6}{\familydefault}{\mddefault}{\updefault}{\color[rgb]{0,0,0}\cond{$(n-1)$}}%
}}}}
\put(12751,-4861){\makebox(0,0)[lb]{\smash{{\SetFigFont{8}{9.6}{\familydefault}{\mddefault}{\updefault}{\color[rgb]{0,0,0}$(n-2)$-process}%
}}}}
\put(12751,-5086){\makebox(0,0)[lb]{\smash{{\SetFigFont{8}{9.6}{\familydefault}{\mddefault}{\updefault}{\color[rgb]{0,0,0}RC is solvable}%
}}}}
\put(7351,-5686){\makebox(0,0)[lb]{\smash{{\SetFigFont{8}{9.6}{\familydefault}{\mddefault}{\updefault}{\color[rgb]{0,0,0}$n$-discerning}%
}}}}
\put(5776,-4861){\makebox(0,0)[lb]{\smash{{\SetFigFont{8}{9.6}{\familydefault}{\mddefault}{\updefault}{\color[rgb]{0,0,0}$n$-process}%
}}}}
\put(5176,-5686){\makebox(0,0)[lb]{\smash{{\SetFigFont{8}{9.6}{\familydefault}{\mddefault}{\updefault}{\color[rgb]{0,0,0}$n$-process consensus}%
}}}}
\put(5551,-5911){\makebox(0,0)[lb]{\smash{{\SetFigFont{8}{9.6}{\familydefault}{\mddefault}{\updefault}{\color[rgb]{0,0,0}is solvable}%
}}}}
\end{picture}%

%% file: hierarchy.tex

\section{Simultaneous Crash Model}
\label{simultaneous}


In the case of simultaneous crashes, the RC hierarchy is identical to the standard consensus hierarchy.

\begin{theorem}
\label{system-thm}
Recoverable consensus is solvable among $n$ processes using objects of type $T$ and registers 
\x{in the simultaneous crash model} if and only if $cons(T)\geq n$.
\end{theorem}

Golab \cite{Gol20} showed how to transform a standard consensus algorithm
into an algorithm for RC in the case of simultaneous crashes.
His transformation required a bound on the number of crashes to ensure
that the space used by the algorithm is bounded.  
Since we allow an unbounded number of objects to be used to solve RC,
a simple modification of Golab's algorithm can be used to prove Theorem \ref{system-thm}.
\arxcam{See Appendix \ref{system-proof} for details.}{See the full version \cite{full} for details.}
In view of Theorem \ref{system-thm}, we focus on determining RC numbers of types in the presence of
\emph{independent} crashes in the remainder of the paper.

%% file: readable-intro.tex

\section{Readable Objects}

A \emph{deterministic} object type has a sequential specification that specifies a unique
response and state transition when a given operation is applied to an object of this type that is in a given state.
An object is \emph{readable} if it has a \op{Read} operation that 
returns the entire state of the object without altering 
it.\footnote{We use this definition for simplicity, but our results would apply equally well to the original, more general definition of readable objects in \cite{Rup00}, which allows the state of the object to be read piece-by-piece.  For example, an array of registers is also readable under the more general definition.}
Ruppert \cite{Rup00} provided a characterization of deterministic, readable types that
can solve consensus among $n$ processes.
In this section, we develop a similar characterization for RC with independent crashes,
and use this to compare the ability of types to solve the two problems.

The characterizations for consensus and for RC
are  linked to the \emph{team consensus} problem, which 
is the problem of solving consensus when the set of processes  are divided in advance into two
non-empty teams and all processes on the same team get the same input.
\x{(This problem is also known as static consensus \cite{Nei95}.)}

We first review the characterization for standard consensus \cite{Rup00}.
Suppose each process can perform a single update operation on an object $O$ of type $T$, and then
read $O$ at some later time, and, based only on the responses of these two steps, determine
which team updated $O$ first.  If this is possible, we say $T$ is $n$-discerning.

\begin{definition}
A deterministic type $T$ is called \emph{$n$-discerning} if there exist 
\begin{itemize}
\item
a state $q_0$, 
\item
a partition of $n$ processes $p_1, \ldots, p_n$
into two non-empty teams $A$ and $B$, and
\item
operations $op_1, op_2, \ldots, op_n$ 
\end{itemize}
such that, for all $j\in\{1,\ldots,n\}$, $R_{A,j} \cap R_{B,j} = \emptyset$, 
where $R_{X,j}$ is the set of pairs $(r,q)$ for which there exist distinct process indices
$i_1, \ldots, i_\alpha$ including $j$ with $p_{i_1}\in X$ such that
if $op_{i_1},\ldots, op_{i_\alpha}$ are performed in this order
on an object of type $T$ initially in state $q_0$, then $op_j$ returns $r$ and the object ends up in state $q$.
\end{definition}

\new{In this definition and in Definition \ref{cond-def}, an operation $op_i$ includes the name
of the operation and any arguments to it.  \x{For example, {\sc Write}(42) is} an operation
on a read/write register.  Operations $op_1, \ldots, op_n$ need not be distinct.}
Ruppert used a valency argument to show that any deterministic, readable type that can solve consensus
among $n$ processes must be $n$-discerning.  Conversely, 
team consensus can be solved using a readable
$n$-discerning object $O$ and one register per team as follows.
Each process \x{$p_i$} writes its input in its team's register, performs \x{its operation $op_i$} on $O$ and then 
reads $O$'s state.  The process determines which team updated $O$ first and outputs
the value  in that team's register.  A tournament
then solves consensus:  processes within each team agree on an input value
recursively and then run team consensus to choose the final output value.
The argument sketched here yields the following characterization.

\begin{theorem}[\cite{Rup00}]
\label{cons-char}
A deterministic, readable type can be used, together with registers, to solve $n$-process
wait-free consensus if and only if it is $n$-discerning.
\end{theorem}

We now consider how to characterize readable types that can solve 
\emph{recoverable} consensus, with independent process crashes.
\new{\emph{Recoverable team consensus} is the RC problem where the processes
are partitioned in advance into two non-empty teams and 
inputs are constrained so that all processes on the same team have the same input value.}
We shall show that RC is solvable if and only if recoverable team consensus is solvable:
\x{the only if direction is trivial, and the converse will be shown
using the same tournament algorithm outlined above.}
So, it suffices to characterize types that can solve recoverable team consensus for $n$ processes.

We shall define a property called \cond{$n$} such that a type $T$ satisfying
the property will allow $n$ processes to solve recoverable team consensus in a simple way.
A shared object $O$ of type $T$ is initialized to some state $q_0$.
To solve team consensus using an $n$-discerning type, each process performs a
single operation on $O$ and then reads $O$, and is able to conclude
from the responses to these two steps which team updated $O$ first.
There are two key difficulties when we consider processes that may crash and recover:
(1) \x{if a process crashes after performing its update, thereby losing the response of that update, the process cannot use the response to determine which team won}, 
and (2) a process that recovers should try to avoid performing its update 
on $O$ a second time so that it does not obliterate the evidence of which team updated $O$ first.

To cope with (1), our new \cond{$n$} property should allow a process
to determine which team updated $O$ first  based simply on the state of $O$,
which can be read at any time.  Thus, two sequences of update operations
that start with processes on opposite teams must not take $O$ to the same state.
This is formalized in condition \ref{empty-cond} of Definition \ref{cond-def}, below.

We now consider how to cope with (2).
If  $O$ could never return to its initial state $q_0$,
\x{checking that $O$'s state is $q_0$ before updating $O$ would ensure that no process ever updates $O$ twice.}
(See the code for processes on team $A$ in Figure \ref{team-alg}.)
However, we can solve team consensus under a  weaker condition:
$O$'s state \emph{can} return to $q_0$ after a process from team $A$ updates $O$ first,
\emph{provided that} team $B$ has only one process.
In this case, condition \ref{empty-cond} of Definition \ref{cond-def} implies that the
state cannot return to $q_0$ if a process on team $B$ updates $O$ first.
Processes on team $A$ behave as before, updating $O$ if they find it in state $q_0$.
If $|B|>1$, processes on team $B$ do likewise.
However, if $|B|=1$, the lone process on team $B$ updates $O$ if it finds $O$ in state $q_0$ \emph{and} sees that no process on team $A$ has started  its algorithm:  in this
case it knows that no operation has been performed on $O$, since $O$ can
return to $q_0$ only if a process on team $A$ updated it first.
If the lone process on team $B$ finds that a process on team $A$ has already started, it simply
outputs team $A$'s input value.
(See the code for processes on team $B$ in Figure \ref{team-alg}.)
This motivates condition \ref{init-A-cond} of Definition \ref{cond-def} below.  A symmetric scenario motivates condition~\ref{init-B-cond}.

\x{The approach of having processes on team $B$ defer to team $A$ if they see that a process on team $A$ has started running works only if $|B|=1$:  if the algorithm used this approach with $|B|>1$, one process on team $B$ might start running before any process on team $A$ and later go on to be the first process to update $O$, while another process 
on team $B$ might start after a process on team $A$ has taken steps and defer to team $A$.  In this case, the latter process on team $B$ would
conclude that team $A$ won, while others would conclude that team $B$ won, violating agreement.}

These considerations lead us to formulate the \cond{$n$} property in Definition \ref{cond-def}, which 
uses the following notation.
\new{
Fix a deterministic, readable type $T$.
Let $X$ be a subset \x{of the set of all processes} $\{p_1, \ldots, p_n\}$
and let $op_1, \ldots, op_n$ be operations. 
Let~$q_0$ be a state of type $T$.
Define $Q_X(q_0, op_1, \ldots, op_n)$ to be 
the set of all states $q$ for which there exist distinct process indices $i_1, \ldots, i_\alpha$ with $p_{i_1}\in X$ such that the sequence of operations $op_{i_1},\ldots, op_{i_\alpha}$
applied to an object of type $T$ initially in state $q_0$ leaves the object in state $q$.
We omit the parameters of $Q_X$ when they are clear from  context. 
}

\begin{definition}
\label{cond-def}
A deterministic type $T$ is \cond{$n$} if
there exist
\begin{itemize}
\item
a state $q_0$,
\item
a partition of $n$ processes $p_1, \ldots, p_n$ into two non-empty teams $A$ and $B$, and
\item
operations $op_1, \ldots, op_n$
\end{itemize}
satisfying the following three conditions.
\begin{enumerate}
\item\label{empty-cond}
$Q_A(q_0, op_1, \ldots, op_n) \cap Q_B(q_0, op_1, \ldots, op_n) = \emptyset$.
\item\label{init-A-cond}
$q_0\notin Q_A(q_0, op_1, \ldots, op_n)$ or $|B|=1$.
\item\label{init-B-cond}
$q_0 \notin Q_B(q_0, op_1, \ldots, op_n)$ or $|A|=1$.
\end{enumerate}
\end{definition}

\ignore{
\newcommand{\weak}{D'}
\begin{definition}
We define a property $\weak_n$ to be the same as $\cond{n}$, except without the requirements \ref{init-A-cond} and \ref{init-B-cond}.
\end{definition}
}

We call a type that satisfies this property \cond{$n$} because it \emph{records} in its state 
information about the team that first updates the object, if it is initialized to state $q_0$.

We first prove some simple consequences of Definition \ref{cond-def}.

\begin{observation}
For $n\geq 2$, if a deterministic type is \cond{$n$}, then it is $n\mbox{-discerning}$.
\end{observation}

To see why this is true, we can use the same choice of $A,B,q_0,$ $op_1,\ldots,op_n$ for both definitions.
If, for some $j$, there were an $(r,q) \in R_{A,j}\cap R_{B,j}$ then $q$ would also be in $Q_A\cap Q_B$,
which would violate property \ref{empty-cond} of the definition of \cond{$n$}.  So we can conclude that $R_{A,j}\cap R_{B,j}$ must be empty, as required for the definition of $n$-discerning.

\begin{observation}
For $n\geq 3$, if a deterministic type is \cond{$n$}, then it is \cond{$(n-1)$}.
\end{observation}

If a type satisfies the definition of \cond{$n$} with teams $A$ and $B$, we can omit one process from the larger team
to get a division of $n-1$ processes into non-empty teams $A'$ and $B'$.  We use the
same initial state $q_0$ and assign the same operation to each process to satisfy the definition \cond{$(n-1)$}.

We  now summarize the results about deterministic, readable types that we prove in the remainder of this section.
\x{Theorem \ref{sufficient} shows that any readable type that is \cond{$n$} is capable
of solving RC among $n$ processes.
We prove in Theorem \ref{necessary} that
all types that can solve RC among $n$ processes satisfy the \cond{$(n-1)$} property.
(This is true even if the type is not readable.)
Given a specification of a shared type, it is fairly straightforward to check whether
it is \cond{$n$}.
By determining the maximum $n$ for which a given readable object type $T$ is \cond{$n$},
we can conclude that $rcons(T)$ is either $n$ or $n+1$.}

We also prove  that an $n$-discerning type must be \cond{$(n-2)$} (Theorem \ref{discerning}), but not necessarily \cond{$(n-1)$} (Proposition \ref{discerning-ceg}).
As a corollary of these results, we show that $cons(T)-2 \leq rcons(T) \leq cons(T)$.
These relationships are summarized in Figure \ref{summary}.
\new{In Theorem \ref{robustness}, we also show how the power of a collection of readable types to solve 
RC is related to the power of each type when used in isolation.}

%% file: readable-sufficient.tex

\subsection{Sufficient Condition}
\label{sec:sufficient}

\begin{figure}
{\small
\begin{code}
\firstline
shared variables\nl
\n  Object $O$ of type $T$, initially in state $q_0$\nl
    Registers $R_A$ and $R_B$, initially in state $\bot$\bl\nl
\p\op{Decide}$(v)$ // code for process $p_i$ on team $A$\nl
\n $R_A \leftarrow v$ \llabel{A-write-input}\nl
   $q \leftarrow O$ \llabel{A-read-1}\nl
   if $q=q_0$ then \llabel{A-begin-if}\nl
\n   	apply $op_i$ to $O$ \llabel{A-apply-op}\nl
        $q\leftarrow O$ \llabel{A-read-2}\nl
\p end if \llabel{A-end-if}\nl
   if $q\in Q_A$ then return $R_A$ \llabel{A-switch}\llabel{A-read-A}\nl
   else return $R_B$ \llabel{A-read-B}\nl
   end if\nl
\p end \op{Decide}\bl\nl
\op{Decide}$(v)$ // code for process $p_i$ on team $B$\nl
\n $R_B \leftarrow v$ \llabel{B-write-input}\nl
   $q \leftarrow O$ \llabel{B-read-1}\nl
   if $q=q_0$ then  \llabel{B-test-1}\nl
\n		if $|B|=1$ and $R_A\neq \bot$ then \llabel{B-test-2}\nl
\n          return $R_A$ \llabel{B-read-A-1}\nl
\p      else  \nl
\n   	    apply $op_i$ to $O$ \llabel{B-apply-op}\nl
            $q\leftarrow O$ \llabel{B-read-2}\nl
\p      end if\nl
\p end if\nl
   if $q\in Q_A$ then return $R_A$ \llabel{B-read-A-2}\nl
   else return $R_B$ \llabel{B-read-B}\nl
   end if\nl
\p end \op{Decide}
\end{code}
}
\caption{Algorithm for recoverable team consensus (assuming $q_0\notin Q_B$). \label{team-alg}}
\end{figure}

We use the algorithm in Figure \ref{team-alg} to show that recoverable team consensus can
be solved using a deterministic, readable object $O$ whose type is \cond{$n$}.  The intuition for the algorithm has
already been described above, but we now describe the code in more detail.
The code assumes $q_0\notin Q_B$; if $q_0\in Q_B$, then $q_0\notin Q_A$ and we would reverse the roles of $A$ and $B$ in the code.
Each process first writes its input  in its team's register.
It then reads $O$.  If $O$ is not in the initial state $q_0$, then the process
 determines which team went first based on the state of $O$ 
and returns the value written in that team's register 
(lines \ref{A-read-A}--\ref{A-read-B} and lines \ref{B-read-A-2}--\ref{B-read-B}).
Otherwise, it updates $O$ before reading the state again (lines \ref{A-apply-op}--\ref{A-read-2} and \ref{B-apply-op}--\ref{B-read-2}) to determine which team updated $O$ first.
There is one exception:  if team $B$ has only one process, it yields to team $A$ (line \ref{B-read-A-1}) if it sees
that some process on team $A$ has already written its input value.
This allows for the case where $q_0\in Q_A$ and $|B|=1$:  it could be that a process
on team $A$ updated $O$ first, and then other processes (including the process on team $B$, in a previous \run)
performed updates that returned $O$ to state $q_0$.
In this case, those processes would have output team $A$'s input value,
so we must ensure that the process on team $B$ does not perform its update again, 
since that could cause processes to output team $B$'s input value, violating agreement. 

\new{The next lemma will help us argue that the algorithm behaves correctly
in the tricky case where $q_0\in Q_A$ and $|B|=1$.}

\begin{lemma}
\label{B-steps}
Suppose $q_0,A,B,op_1,\ldots,op_n$ satisfy the definition of \cond{$n$} for a deterministic type $T$.
Let $X\in \{A,B\}$.
If $q_0\notin Q_X$ and $i_1, \ldots, i_\alpha$ is a sequence of distinct process indices such that 
the sequence of operations $op_{i_1}, \ldots, op_{i_\alpha}$ takes an object of type
$T$ from state $q_0$ to state $q_0$, then the indices of 
all processes of team $X$ appear in the sequence.
\end{lemma}
\begin{proof}
To derive a contradiction, suppose the claim is false, i.e., $j\notin\{i_1,\ldots,i_\alpha\}$ for some process $p_j$ on team $X$.
If $p_{i_1}$ were on team $X$, then the fact that the sequence of operations $op_{i_1}, \ldots, op_{i_\alpha}$ take the state of an object from $q_0$ to $q_0$ would imply that $q_0 \in Q_X$, 
contrary to our assumption.
Thus, $p_{i_1}$ must be on the opposite team $\overline{X}$.
Let $q_j$ be the state that results when $op_j$ is applied to an object in state $q_0$.
We have $q_j\in Q_X$ since the sequence $op_j$ takes an object from state $q_0$ to $q_j$.
We also have $q_j\in Q_{\overline{X}}$ since the sequence $op_{i_1}, \ldots, op_{i_\alpha}, op_j$ takes an object of type
$T$ from state $q_0$ back to state $q_0$ and then to state $q_j$.
Thus, $q_j\in Q_X\cap Q_{\overline{X}}$, which violates condition \ref{empty-cond} in the definition of \cond{$n$}.
\end{proof}

\new{ 
To gain some intuition,
we describe why the following bad scenario cannot occur when 
$|B| = 1$ and $q_0 \in Q_A$. 
Suppose a process $p_1$ on team $B$ begins, sees
$R_A  = \bot$, and is poised to update $O$ at line~\ref{B-apply-op}.
Then, a process $p_2$ on team $A$ runs to completion, updating $O$ and deciding $R_A$. 
Then, other processes update $O$, returning $O$'s state to $q_0$. 
If $p_1$ were still poised to update $O$ at line \ref{B-apply-op},
then it would decide $R_B$, violating agreement.
But this cannot happen:  Lemma \ref{B-steps} ensures that $p_1$ must have been among the
processes that already applied their operations on $O$ to return $O$'s state to $q_0$.

We also describe why the condition $|B|=1$ on line \ref{B-test-2} is necessary.
If this test were missing, consider an execution where one process $p_1$ on team $B$ begins, sees $R_A=\bot$
and is about to update $O$ at line \ref{B-apply-op}.
Then, a process $p_2$ on team $A$ writes to $R_A$.
Next, another process $p_3$ on team $B$ sees that $R_A\neq \bot$ and decides $R_A$ 
(at line~\ref{B-read-A-1}).
Finally, process $p_1$  resumes  and updates $O$.
Since it is the first process to update $O$, $O$'s state would then be in $Q_B$, so
$p_1$ would then read $O$ and decide $R_B$, violating agreement.
We avoid this scenario  by the test $|B| = 1$ of line~\ref{B-test-2}: line~\ref{B-read-A-1} is executed only if
$B$ contains just one process (whereas two processes on team $B$ are needed for the bad scenario described above).
}

\begin{theorem}
\label{sufficient}
If a deterministic, readable type $T$ is \cond{$n$}, then
objects of type $T$, together with registers, can be used to solve recoverable consensus for $n$ processes.
\end{theorem}

\begin{proof}
If team recoverable consensus can be solved, then RC can be solved.
Processes on each team agree recursively on an input value for their team, and then use team consensus
to determine the final output.  See \arxcam{Appendix \ref{team-tournament}}{the full version \cite{full}} for details.

Thus, it suffices to show that the algorithm in Figure \ref{team-alg} 
solves recoverable team consensus using a type $T$ that satisfies the condition of the theorem.
Since $Q_A \cap Q_B = \emptyset$, we know that either $q_0\notin Q_A$ or $q_0\notin Q_B$.
Without loss of generality, assume $q_0\notin Q_B$.  (If this is not the case, just swap the names of the two teams.)

Recoverable wait-freedom is clearly satisfied, since there are no loops in the code.
It remains to show that every execution of the algorithm satisfies validity and agreement.

\begin{lemma}
Validity and agreement are satisfied in executions where no process ever performs an update  on $O$.
\end{lemma}
\begin{proof}
In this case, $O$ remains in state $q_0$ forever.
Thus, no process can reach line \lref{A-read-A} or \lref{B-read-A-2}, since it would first have to 
update~$O$ at line \lref{A-apply-op} or \lref{B-apply-op}, respectively.
So, \x{processes output  only at line~\lref{B-read-A-1}.}
By the test on line \lref{B-test-2}, $R_A$ is written before a process outputs its value on line \ref{B-read-A-1}.
Thus, all outputs are the input value of team~$A$.
\end{proof}

For the remainder of the proof of the theorem, consider executions where at least one update is performed on $O$.
Let $s$ be the first step in the execution that performs an update on $O$.

\begin{lemma}
\label{states-lem-1}
For $X\in\{A,B\}$,
if a process on team $X$ performs $s$  and $q_0\notin Q_X$, then $O$'s state is in $Q_X$ at all times after $s$.
\end{lemma}

\begin{proof}
We first show that no process performs more than one update on~$O$.
To derive a contradiction, suppose some process performs two updates on~$O$.
Let $s'$ be the first step \x{in the execution when a process performs its {\it second} 
update on~$O$ and let $p_i$ be the process that performs $s'$}.
Let $r'$ be $p_i$'s \run\ of the code that performs~$s'$.
Since $r'$ begins after $p_i$'s first update on $O$, $r'$ begins after~$s$.
By definition of $s'$, each process does at most one update on $O$ before~$s'$.
Thus, the state of $O$ is in $Q_X$ at all times between $s$ and~$s'$.
Since $q_0\notin Q_X$, the state of $O$ is never $q_0$ between $s$ and~$s'$. 
This contradicts the fact that $r'$ must read the state of $O$ to be $q_0$ between $s$ and $s'$; 
otherwise $r'$ would not perform~$s'$.

Thus, each process performs at most one update on $O$.  By the definition of $Q_X$,
the state of $O$ is in $Q_X$ at all times after $s$.
\end{proof}

\y{
We next prove a similar lemma for the case where $q_0\in Q_A$.
}
In this case, the situation is a little more complicated.
The state of $O$ might return to $q_0$.  If this happens, we show that each process updates $O$ at most
once before the state returns to $q_0$, and that only processes of team $A$ can update $O$
after the state returns to $q_0$ and each process does so at most once.  
This is enough to ensure that $O$'s state remains
in $Q_A$ at all times.

\begin{lemma}
\label{states-lem-2}
If $s$ is performed by a process of team $A$ and $q_0\in Q_A$, 
then $O$'s state is in $Q_A$ at all times.
\end{lemma}

\begin{proof}
Since $q_0\in Q_A$, there is a unique process $p_j$ on team $B$, by condition \ref{init-A-cond} of the definition of \cond{$n$}.
$O$'s state is $q_0\in Q_A$ at all times before $s$.  
It remains to show that $O$'s state is in $Q_A$ at all times after $s$.
We consider two~cases.

First, suppose $O$ is never in state $q_0$ after $s$.  Consider any process $p_i$ that 
performs an update on $O$.  Let $s_i$ be $p_i$'s first update on $O$.  
\new{By definition, $s_i$ is either equal to $s$ or after $s$.}
Any run by $p_i$ that begins after $s_i$ (and hence after~$s$) 
that reads $O$ on line \lref{A-read-1} 
or \lref{B-read-1} sees a value different from $q_0$, so it does not perform an update on $O$.
Thus, no process performs more than one update on $O$.  
It follows from the definition of $Q_A$ that $O$'s state is in $Q_A$
at all times after $s$.

Now, suppose $O$'s state is equal to $q_0$ at some time after $s$.
Let $s''$ be the first step at or after $s$ that changes $O$'s state back to $q_0$.
\new{We next prove that no process performs two updates on $O$ between $s$ and $s''$ (inclusive).
To derive a contradiction, suppose some process performs two such updates.}
Let $s'$ be the first step when any process performs its {\it second} 
update on $O$.
\new{By definition, $s'$ is between $s$ and $s''$ (inclusive).}
Let $p_i$ be the process that performs $s'$ and let 
$r'$ be the \run\ by $p_i$ that performs $s'$.
Since $r'$ begins after $p_i$'s first update to $O$, $r'$ begins after $s$.
Thus, $r'$ reads $O$'s state to be different from $q_0$ at line \lref{A-read-1} or \lref{B-read-1}, and therefore fails the test on line \lref{A-begin-if} or \lref{B-test-1}.
This contradicts the fact that $r'$ updates $O$.
Hence, each process performs at most one update on $O$ between $s$ and $s''$ (inclusive).

It follows from the definition of $Q_A$ that the state of $O$ is in $Q_A$ at all times between $s$ and $s''$.
By Lemma \ref{B-steps}, the unique process $p_j$ on team $B$ 
updates $O$ between $s$ and $s''$ (inclusive).

Next, we argue that the process $p_j$ on team $B$ updates $O$ exactly once in the entire execution.
We have already seen that $p_j$ updates $O$ exactly once between $s$ and $s''$ (inclusive).
Any \run\ by process $p_j$ that begins after that first update to $O$ by $p_j$
(and therefore after $s$) would see that $R_A\neq \bot$, since the process on
team $A$ that performs $s$ writes to $R_A$ before~$s$.
That \run\ by $p_j$ would therefore pass the test on line \lref{B-test-2}
and could not update $O$ on line~\lref{B-apply-op}.

Thus, any updates to $O$ after $s''$ are by processes in $A$.
If there are no updates to $O$ after $s''$, then $O$ remains in state $q_0\in Q_A$ at all times after $s''$.
If there is some update to $O$ after $s''$, let $s'''$ be the first one.
\x{Since $q_0\notin Q_B$ and no process on team $B$ updates $O$ after $s''$,
 the state of $O$ can never be $q_0$ after $s'''$, by Lemma \ref{B-steps}.}
Consider any process $p_i$ on team $A$ that performs an update on $O$ after $s''$.
Let $s_i$ be $p_i$'s first update on $O$ after $s''$.
\x{By the definition of $s'''$, $s_i$ is either $s'''$ or after $s'''$.}
Any \run\ by $p_i$ that begins after $s_i$ (and therefore after $s'''$)
that reads $O$ on line~\lref{A-read-1} will see a value different from $q_0$,
so it does not perform an update on~$O$.
Thus, no process performs more than one update on $O$ after $s''$.
It follows from the definition of $Q_A$ that $O$'s state is in $Q_A$ at all times after~$s''$.
\end{proof}

\begin{lemma} 
\label{A-output}
Any output produced by a process on team $A$ is the input value of
the team that first updated $O$.
\end{lemma}
\begin{proof}
Consider a run $r$ of the code by a process in $A$ that produces an output.
If $r$ reads $O$ at line \lref{A-read-1} before $s$, then it will read the value $q_0$
and read $O$ again at line \lref{A-read-2}, which is after $s$.
Thus, the value tested at line \lref{A-switch} is read from $O$ after~$s$.

If the first update to $O$ is by a process on team $A$, the value tested is in $Q_A$,
by Lemma \ref{states-lem-1} and \ref{states-lem-2}.
So, $r$ outputs the value of $R_A$.

If the first update to $O$ is by a process on team $B$, the value tested is in $Q_B$,
by Lemma \ref{states-lem-1} and the fact that $q_0\notin Q_B$.  
Since $Q_A\cap Q_B=\emptyset$, the value tested will not be in $Q_A$.
So, $r$ outputs the value of $R_B$.

In both cases, the relevant register is written before $s$, so $r$ outputs the input value of the team that first updates~$O$.
\end{proof}

\begin{lemma} 
\label{B-output}
Any output produced by a process on team $B$ is the input value of
the team that first updated $O$.
\end{lemma}
\begin{proof}
Consider any run $r$ of the code by a process in $B$ that produces an output.
We consider three cases.

\begin{enumerate}[{Case} 1:]
\item a process from team $A$ performs $s$.

We first show $r$ returns a value read from $R_A$ by considering two subcases.
\begin{enumerate}[(a)]
\item
$q_0\in Q_A$.
In this case $|B|=1$, by condition \ref{init-A-cond} of the definition of \cond{$n$}.
By Lemma \ref{states-lem-2}, $O$'s state is in $Q_A$ at all times,
so $r$ cannot return at line \lref{B-read-B}.
Therefore, $r$ outputs the value it reads from $R_A$ at line \lref{B-read-A-1} or~\lref{B-read-A-2}.
\item
$q_0\notin Q_A$.
By Lemma \ref{states-lem-1}, $O$'s state is in $Q_A$ at all times after $s$.
If $r$ reads $O$ at line \lref{B-read-1} before $s$, it will see $q_0$ and execute the
test at line \lref{B-test-2}.
Then, it will either return the value in $R_A$ at line \lref{B-read-A-1},
or read $O$ again at line \lref{B-read-2} after $s$, getting a value in $Q_A$ and returning
 the value in $R_A$ at line \lref{B-read-A-2}.
\end{enumerate}

To derive a contradiction, suppose $R_A$ is still $\bot$ when $r$ reads it at line \lref{B-read-A-1} or \lref{B-read-A-2}.
Then, $r$ returns before $s$, since $R_A$ must be written before~$s$.
So $r$ must have read $q_0$ from $O$ at line \lref{B-read-1}.
Thus, the test at line \lref{B-test-1} is true and the test at line \lref{B-test-2} is false,
so $r$ performs an update on $O$ before $s$, contradicting the definition of $s$.

Therefore, $r$ outputs team $A$'s input value, as required.

\item
A process from team $B$ performs $s$ and $|B|>1$.
Since $q_0\notin Q_B$, it follows from Lemma \ref{states-lem-1} that
$O$'s state is in $Q_B$ at all times after $s$.
If $r$ reads $O$ at line \lref{B-read-1} before $s$, it will see $q_0$ and execute the
test at line \lref{B-test-2}, which fails because $|B|>1$.
Then, it will read $O$ again at line \lref{B-read-2} after $s$, getting a value in $Q_B$ and return
 the value in $R_B$ at line \lref{B-read-B}.
Since $r$ wrote $R_B$ at line \lref{B-write-input}, $r$ outputs team $B$'s input value, as required.

\item
A process from team $B$ performs $s$ and $|B|=1$.
Let $p_j$ be the unique process on  team $B$.
By Lemma \ref{states-lem-1} and the fact that $q_0\notin B$, the state of $O$ is in $Q_B$ at all times after $s$.

If $r$ is the \run\ of $p_j$ that performs $s$, then $r$ sees $R_A=\bot$ on line \lref{B-test-2};
otherwise it would not execute line \lref{B-apply-op}.
So, if $r$ returns a value, it reads $O$ at line \lref{B-read-2} after $s$ and gets a value in $Q_B$.
It  must then return a value at line \lref{B-read-B}.

Any \run\ $r$ of $p_j$ that ends before $s$ 
 evaluates the test at line \lref{B-test-1} to true
and the test at line \lref{B-test-2} to false, so it must crash before reaching line \lref{B-apply-op} and  does not produce an output.

If $r$ is a run of $p_j$ that starts after $s$, it  reads a value in $Q_B$ at line \lref{B-read-1}.
Since $q_0 \notin Q_B$, it would return at line~\lref{B-read-B}.

Thus, all outputs by $p_j$ are  read from $R_B$ at line \ref{B-read-B}, which contains
team $B$'s input value  written at line \lref{B-write-input}.
\end{enumerate}
\vspace*{-6mm}
\end{proof}

Lemmas \ref{A-output} and \ref{B-output} prove validity and agreement when some process updates $O$, completing the proof of Theorem~\ref{sufficient}.
\end{proof}

%% file: readable-necessary.tex

\subsection{Necessary Condition}
\label{sec:necessary}

\new{In this section, we show that being \cond{$(n-1)$} is a necessary condition for a deterministic type to be capable of solving $n$-process
RC.  This result holds whether the type is readable or not.
The proof uses a valency argument~\cite{FLP85}.}
Assuming an algorithm exists, the valency argument constructs an infinite execution
in which no process ever returns a value.
Unfortunately, in the case of RC, it is possible to have an infinite
execution where no process returns a value (if infinitely many crashes occur).
Thus, the proof considers a restricted set of executions where each execution must produce an output value for some process within 
a finite  number of steps, and uses this restricted set to define valency.
This technique was used by Golab \cite{Gol20} to prove a necessary condition (weaker than the 2-recording property) for solving 2-process RC.  Lo and Hadzilacos \cite{LH00} had previously used
a similar technique of defining valency using a pruned execution tree.
\new{Attiya, Ben-Baruch and Hendler \cite{ABH18} also used a valency argument in the context
of non-volatile memory in their proof that a recoverable test-and-set object cannot be built
from ordinary test-and-set objects (and registers).}

\begin{theorem}
\label{necessary}
For $n\geq 3$, if a deterministic type $T$ can be used, together with registers, to solve recoverable consensus among $n$ processes, then $T$ is \cond{$(n-1)$}.
\end{theorem}

\begin{proof}
Assume there is an algorithm $A$ for RC among $n$ processes $p_1, \ldots, p_n$ using objects of type $T$ and registers.
Let ${\mathcal E_A}$ be the set of all executions of $A$ where
$p_2, \ldots, p_n$ never crash, and
in any prefix of the execution, the number of crashes of $p_1$ is less than or equal to the total number of steps of $p_2, \ldots, p_n$.

Consider a finite execution $\gamma$ in ${\mathcal E}_A$.
Define $\gamma$ to be \emph{$v$-valent} if there is no decision 
different from $v$ in any extension of $\gamma$ in ${\mathcal E}_{A}$.
An execution $\gamma$ cannot be both $v$-valent \emph{and} $v'$-valent 
if $v\neq v'$, since a failure-free extension of $\gamma$ must eventually produce a decision.
We call $\gamma$ \emph{univalent} if it is $v$-valent for some $v$,
or \emph{multivalent} otherwise.

To see that a multivalent execution exists,
consider an execution with no steps where processes $p_1$ and $p_2$ have inputs 0
and 1.  If $p_1$ runs by itself, it must  output 0;
if $p_2$ runs by itself it must  output 1.

Next, we argue that there is a \emph{critical execution} $\gamma$, i.e., a multivalent execution
in ${\mathcal E}_A$ such that every extension of $\gamma$ in ${\mathcal E}_{A}$ is univalent.
If there were not, we could construct an infinite execution
of ${\mathcal E}_{A}$ in which every prefix is multivalent, meaning that no process ever returns a value.
Such an execution could be constructed inductively by starting with a multivalent execution and, at each step of the induction, extending it to a longer multivalent execution.
This would violate the termination property of RC, since some process takes an infinite number of steps without crashing.

For $1\leq i\leq n$, let $v_i$ be the 
value such that $\gamma$ followed by the next step of $p_i$'s 
algorithm is $v_i$-valent.
We show not all of $v_1, \ldots, v_n$ are the same.
To derive a contradiction, suppose they are all equal.
Since $\gamma$ is multivalent, some extension of $\gamma$ in ${\mathcal E}_A$ is $v'$-valent for some $v'\neq v_2$.  
\x{By assumption, the next step of each process's  algorithm produces a $v_2$-valent execution,
so the $v'$-valent extension must begin with a crash of $p_1$.}
But the extensions of $\gamma$ shown in Figure \ref{necessary-fig2}(a) are indistinguishable
to  $p_2$.
Thus, $p_2$  returns the same value in both, contradicting the
fact that one extends a $v_2$-valent execution and the other
extends a $v'$-valent execution, where $v_2\neq v'$.

\begin{figure}
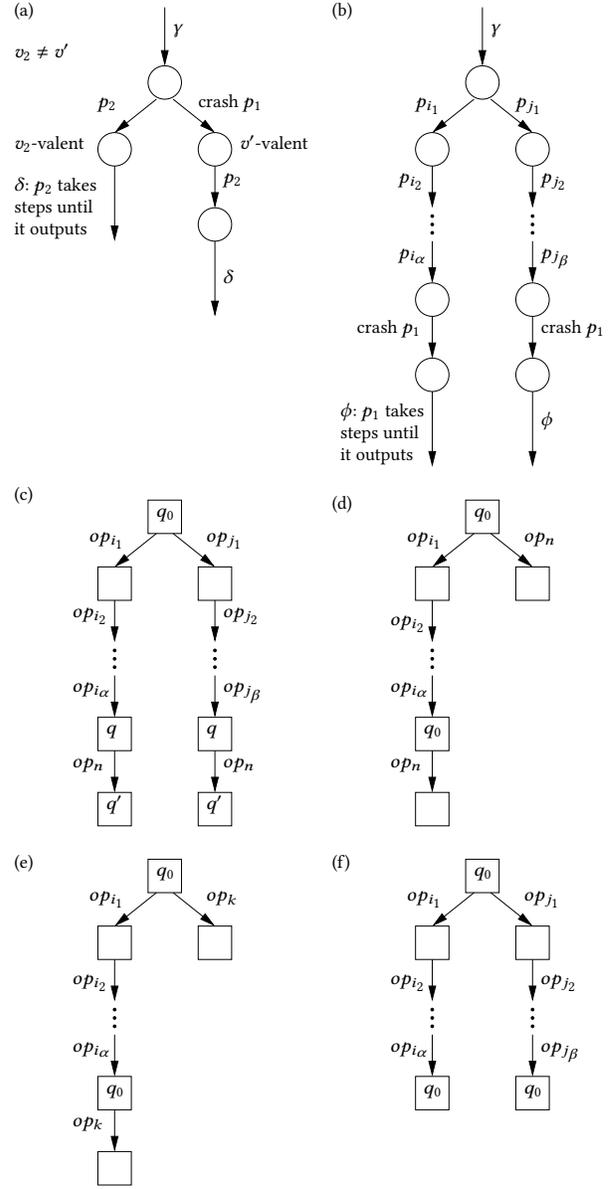
\Description{This figure has 6 parts.  
Part a shows two extensions of $\gamma$ that $p_2$ cannot distinguish:  in one $p_2$ runs solo to completion, and in the other $p_1$ crashes and then $p_2$ runs solo to completion.
Part b shows two extensions of $\gamma$ that $p_1$ cannot distinguish:  in one, $p_{i_1}, \ldots, p_{i_\alpha}$ each take a step and then $p_1$ crashes, recovers and runs to completion, and in the other, $p_{j_1},\ldots, p_{j_2}$ each take a step and then $p_1$ crashes, recovers and runs to completion.
Part c shows two sequences of operations that take the object from state $q_0$ to the same state $q'$;
the sequences are $op_{i_1}, op_{i_2}, \ldots, op_{i_\alpha}, op_n$ and $op_{j_1}, op_{j_2},\ldots, op_{j_\beta}, op_n$.
Part d shows two sequences of operations on an object initially in state $q_0$;
the sequences are $op_{i_1}, op_{i_2},\ldots, op_{i_\alpha}, op_n$ and $op_n$.
Part e shows two sequences of operations on an object initially in state $q_0$;
the sequences are $op_{i_1}, op_{i_2},\ldots, op_{i_\alpha}, op_k$ and $op_k$.
Part f shows two sequences of operations that take the object from state $q_0$ back to state $q_0$;
the sequences are $op_{i_1}, op_{i_2},\ldots, op_{i_\alpha}$ and $op_{j_1}, op_{j_2},\ldots, op_{j_\beta}$.
}
\caption{Proof of Theorem \ref{necessary}.  Circles represent states of the system.  Squares represent the state of $O$.
\label{necessary-fig2}}
\end{figure}

A standard argument shows that at the end of $\gamma$, each process is about to perform
an operation on the same object $O$ of type $T$, and that step cannot be a read operation.
For $i\in \{1,\ldots,n\}$, 
let $op_i$ be the update operation that $p_i$ is poised to perform on $O$ after $\gamma$.
Let $q_0$ be the state of $O$ at the end of $\gamma$.

We next prove a technical lemma that will be used several times to complete the theorem's proof.
\new{It captures a  valency argument  we use:  
if two sequences of steps by distinct processes chosen from $p_1, \ldots, p_n$ 
after $\gamma$ can take $O$ to the same state 
\emph{and} process $p_1$ can crash after \x{both} of them, then the two extensions
must have the same valency.  To ensure that $p_1$ can crash, the hypothesis of the lemma requires that neither sequence consists of a single step
by~$p_1$.}

\begin{lemma}
\label{same-valency}
Suppose there is a sequence of distinct process ids $i_1, \ldots, i_\alpha$ 
and another sequence of distinct ids $j_1, \ldots, j_\beta$ such that
each sequence contains an element of $\{2, \ldots, n\}$ and
the sequences of operations 
$op_{i_1}, \ldots, op_{i_\alpha}$ and $op_{j_1}, \ldots, op_{j_\beta}$ both take
object $O$ from state $q_0$ to the same state $q$.
Then, $v_{i_1} = v_{j_1}$.
\end{lemma}
\begin{proof}
The two executions in Figure \ref{necessary-fig2}(b) are in ${\mathcal E}_A$ since
one of $p_2, \ldots, p_n$ takes a step in each extension of $\gamma$ before $p_1$ crashes.
$O$ is in state $q$ before $p_1$ crashes in both extensions, and no other shared object
changes between the end of $\gamma$ and the crash of $p_1$.
Thus, these two extensions are indistinguishable to the last \run\ $\phi$ of $A$ by $p_1$.
Since the left extension is $v_{i_1}$-valent and the right extension is $v_{j_1}$-valent,
we must have $v_{i_1}=v_{j_1}$.
\end{proof}

We now describe how to split $n-1$ of the processes into two teams $A$ and $B$ according to their valency 
to satisfy the definition of \cond{$(n-1)$}.
\new{The following two cases describe how to relabel the processes (if necessary) so that we can
split processes $p_1, \ldots, p_{n-1}$ into the two required teams.}

\begin{enumerate}[{Case} 1:]
\item
\label{single}
Suppose there is an $i$ such that, for all $j\neq i$, $v_i\neq v_j$.
Without loss of generality, assume that $i<n$.
(If $i=n$, we can swap the ids of $p_2$ and $p_n$ to ensure $i<n$, since $n\geq 3$.)
Let $A=\{p_i\}$ and $B=\{p_1, \ldots, p_{n-1}\}-\{p_i\}$.
\item
\label{no-singles}
Suppose that for every $i$, there is a $j\neq i$ such that $v_i=v_j$.
If there is a sequence of distinct ids $i_1, \ldots, i_\alpha$ chosen from $\{1,\ldots,n\}$ such that
the sequence of operations $op_{i_1}, \ldots, op_{i_\alpha}$
take the object $O$ from state $q_0$ back to $q_0$, then let $\ell=i_1$.
Otherwise, let $\ell$ be any id.
Without loss of generality, assume $\ell < n$.  (If this is not the case, swap the labels of processes $n-1$ and $n$ to make it true.)
Again, without loss of generality, assume $v_n \neq v_\ell$.
(Since not all of $v_1, \ldots, v_n$ are the same, there is some $\ell'$ such that 
$v_{\ell'} \neq v_\ell$.
By the assumption of Case \ref{no-singles}, we can choose such an $\ell' > 1$.
If $\ell' < n$, swap the ids of $p_{\ell'}$ and $p_n$.  This ensures that $v_n \neq v_\ell$.)

Then, define $A$ to be $\{p_i : 1\leq i\leq n-1 \mbox{ and } v_i=v_\ell\}$
and $B$ to be $\{p_i : 1\leq i \leq n-1 \mbox{ and } v_i\neq v_\ell\}$.
It follows from the fact that not all of $v_1, \ldots, v_n$ are the same and the assumption of 
Case \ref{no-singles}, that both teams are non-empty.
\end{enumerate}

It follows from the definitions of $A$ and $B$ that, in either case, \new{they form a partition of the 
processes $p_1, \ldots, p_{n-1}$ into two non-empty teams satisfying} the following properties:
\begin{enumerate}[P1:]
\item
\label{team-valencies}
$v_i \neq v_j$ for all $p_i\in A$ and $p_j\in B$, and
\item
\label{omitted-valency}
$v_i \neq v_n$ for all $p_i\in A$.
\end{enumerate}

We  check that $Q_A(q_0,op_1,\ldots,op_{n-1})$ and $Q_B(q_0,op_1,\ldots,op_{n-1})$ satisfy the definition of \cond{$(n-1)$}.

To derive a contradiction, suppose there is a state $q\in Q_A\cap Q_B$.
This means there is a sequence of distinct process ids $i_1, \ldots, i_\alpha$ chosen from $\{1,\ldots,n-1\}$ with $p_{i_1}\in A$ and
another sequence of distinct process ids $j_1, \ldots, j_\beta$ chosen from $\{1,\ldots,n-1\}$ with $p_{j_1}\in B$ such that 
the sequences $op_{i_1},\ldots, op_{i_\alpha}$ and $op_{j_1},\ldots, op_{j_\beta}$ both take
object $O$ from state $q_0$ to state $q$.
Adding one more operation $op_n$ to the end of these sequences would leave
$O$ in the same state $q'$.  (See Figure~\ref{necessary-fig2}(c).)
By Lemma \ref{same-valency}, $v_{i_1} = v_{j_1}$.
This contradicts property P\ref{team-valencies}.
Thus, condition \ref{empty-cond} of the definition of \cond{$(n-1)$}~holds.

To derive a contradiction, suppose $q_0\in Q_A$.
Then, there is a sequence of distinct process ids $i_1, \ldots, i_\alpha$ chosen from $\{1,\ldots,n-1\}$
with $p_{i_1}\in A$ such that
the sequence of operations $op_{i_1}, \ldots, op_{i_\alpha}$ takes
object $O$ from state $q_0$ back to state~$q_0$.
The two sequences of operations on $O$ shown in Figure \ref{necessary-fig2}(d) both leave
$O$ in the same state.
Thus, $v_{i_1} = v_n$,
by Lemma \ref{same-valency},  contradicting property P\ref{omitted-valency}.
Thus, condition \ref{init-A-cond} of the definition of \cond{$(n-1)$} is satisfied.

To derive a contradiction, suppose $q_0 \in Q_B$ and $|A|>1$.
Since $|A|>1$, the teams must have been defined according to Case \ref{no-singles}.
Since $q_0\in Q_B$, there is a sequence of distinct process ids 
$j_1, \ldots, j_\beta$ chosen from $\{1, \ldots, n-1\}$ with $p_{j_1}\in B$ such that $op_{j_1}, \ldots, op_{j_\beta}$ takes 
object $O$ from state $q_0$ back to $q_0$.
\x{So, in Case \ref{no-singles} of the definition of the teams,
we chose $\ell=i_1$, where $i_1,\ldots,i_\alpha$ is some sequence of distinct process ids chosen from
$\{1,\ldots,n\}$ such that 
$op_{i_1}, \ldots, op_{i_\alpha}$  also takes object $O$ from state $q_0$ back to $q_0$.}
Since $i_1=\ell\leq n-1$, we have $p_{i_1}\in A$.
(We remark that  this sequence's existence does not contradict the fact proved above that $q_0\notin Q_A\x{(q_0,op_1,\ldots,op_{n-1})}$, since this sequence may include the index~$n$.)

Our goal is to show that $v_{i_1}=v_{j_1}$, which will contradict property P\ref{team-valencies}.
We use a case argument, showing that it is possible to apply Lemma \ref{same-valency} in each case.
Let $I =\{k : 2\leq k\leq n \mbox{ and } v_k = v_{i_1}\}$ and let
$J = \{k : 2\leq k \leq n \mbox{ and } v_k = v_{j_1}\}$.
\new{A step by a process whose index is in $I$ or $J$ extends the critical execution $\gamma$
to a $v_{i_1}$- or $v_{j_1}$-valent execution, respectively.  Moreover, a step by any process
in $I$ or $J$ allows
us to invoke Lemma \ref{same-valency} since the sets $I$ and $J$ do not include~1.}

\begin{enumerate}[{Case} a:]
\item
\label{J-not-used}
Suppose  some $k\in J$  does not appear in $i_1, \ldots, i_\alpha$.
Then, the two sequences of operations on $O$  in Figure \ref{necessary-fig2}(e)
leave $O$ in the same state.
Since $k\geq 2$,  Lemma \ref{same-valency} implies that $v_{i_1} = v_k$.
By definition of $J$, $v_k=v_{j_1}$.  Thus,~$v_{i_1}=v_{j_1}$.
\item
Suppose there is some $k\in I$ that does not appear in $j_1, \ldots, j_\beta$.
By an argument symmetric to Case \ref{J-not-used}, $v_{i_1}=v_{j_1}$.
\item
Suppose $J\subseteq \{i_1, \ldots, i_\alpha\}$ and $I\subseteq \{j_1, \ldots, j_\beta\}$.
We first argue that $I$ is non-empty.
If $i_1 > 1$, then $i_1\in I$.
Otherwise, $i_1=1$ and by the assumption of Case 2, there is some other process id $k$
such that $v_k=v_{i_1}$ and this $k$ is in $I$.
A symmetric argument can be used to show that $J$ is non-empty.
Thus, both of the sequences $i_1, \ldots, i_\alpha$ and $j_1, \ldots, j_\beta$
contain at least one of the ids in $\{2, \ldots, n\}$.
Since both sequences of operations shown in Figure \ref{necessary-fig2}(f) 
leave $O$ in the same state $q_0$, it follows from Lemma \ref{same-valency} that $v_{i_1} = v_{j_1}$.
\end{enumerate}
In all three cases, $v_{i_1} = v_{j_1}$, contradicting Property P\ref{team-valencies}.
Thus, condition \ref{init-B-cond}  of the definition of \cond{$(n-1)$} holds.
\end{proof}

In proving that $T$ is \cond{$(n-1)$}, we split $n-1$ of the processes
into two teams according to the valency induced by their next step after the critical execution
and assigned each process the operation they perform in this  step.
To show that these choices satisfy the definition of \cond{$(n-1)$}, it was essential
to have one process $p_n$ ``in reserve'' that we could use to take one step in Figures 
\ref{necessary-fig2}(c) and \ref{necessary-fig2}(d).  This step enables the crash of $p_1$
needed to prove Lemma \ref{same-valency}, which shows that the two executions in those
figures lead to the same outcome, thereby deriving the necessary contradiction.
This is the reason we show that being \cond{$(n-1)$} (rather than \cond{$n$}) is necessary for solving RC.

%% file: necessary-fig.pdf_t
\begin{picture}(0,0)%
\includegraphics{necessary-fig.pdf}%
\end{picture}%
\setlength{\unitlength}{2763sp}%
\begingroup\makeatletter\ifx\SetFigFont\undefined%
\gdef\SetFigFont#1#2#3#4#5{%
  \reset@font\fontsize{#1}{#2pt}%
  \fontfamily{#3}\fontseries{#4}\fontshape{#5}%
  \selectfont}%
\fi\endgroup%
\begin{picture}(5727,10635)(886,-14023)
\put(901,-11161){\makebox(0,0)[lb]{\smash{{\SetFigFont{7}{8.4}{\familydefault}{\mddefault}{\updefault}{\color[rgb]{0,0,0}(e)}%
}}}}
\put(901,-5281){\makebox(0,0)[lb]{\smash{{\SetFigFont{7}{8.4}{\familydefault}{\mddefault}{\updefault}{\color[rgb]{0,0,0}steps until}%
}}}}
\put(901,-5476){\makebox(0,0)[lb]{\smash{{\SetFigFont{7}{8.4}{\familydefault}{\mddefault}{\updefault}{\color[rgb]{0,0,0}it outputs}%
}}}}
\put(3826,-7306){\makebox(0,0)[lb]{\smash{{\SetFigFont{7}{8.4}{\familydefault}{\mddefault}{\updefault}{\color[rgb]{0,0,0}steps until}%
}}}}
\put(3826,-7501){\makebox(0,0)[lb]{\smash{{\SetFigFont{7}{8.4}{\familydefault}{\mddefault}{\updefault}{\color[rgb]{0,0,0}it outputs}%
}}}}
\put(3826,-7111){\makebox(0,0)[lb]{\smash{{\SetFigFont{7}{8.4}{\familydefault}{\mddefault}{\updefault}{\color[rgb]{0,0,0}$\phi$: $p_1$ takes}%
}}}}
\put(5176,-3661){\makebox(0,0)[lb]{\smash{{\SetFigFont{7}{8.4}{\familydefault}{\mddefault}{\updefault}{\color[rgb]{0,0,0}$\gamma$}%
}}}}
\put(5401,-4336){\makebox(0,0)[lb]{\smash{{\SetFigFont{7}{8.4}{\familydefault}{\mddefault}{\updefault}{\color[rgb]{0,0,0}$p_{j_1}$}%
}}}}
\put(4501,-4336){\makebox(0,0)[lb]{\smash{{\SetFigFont{7}{8.4}{\familydefault}{\mddefault}{\updefault}{\color[rgb]{0,0,0}$p_{i_1}$}%
}}}}
\put(4351,-5011){\makebox(0,0)[lb]{\smash{{\SetFigFont{7}{8.4}{\familydefault}{\mddefault}{\updefault}{\color[rgb]{0,0,0}$p_{i_2}$}%
}}}}
\put(5626,-5011){\makebox(0,0)[lb]{\smash{{\SetFigFont{7}{8.4}{\familydefault}{\mddefault}{\updefault}{\color[rgb]{0,0,0}$p_{j_2}$}%
}}}}
\put(4351,-5686){\makebox(0,0)[lb]{\smash{{\SetFigFont{7}{8.4}{\familydefault}{\mddefault}{\updefault}{\color[rgb]{0,0,0}$p_{i_\alpha}$}%
}}}}
\put(5626,-5686){\makebox(0,0)[lb]{\smash{{\SetFigFont{7}{8.4}{\familydefault}{\mddefault}{\updefault}{\color[rgb]{0,0,0}$p_{j_\beta}$}%
}}}}
\put(3976,-6361){\makebox(0,0)[lb]{\smash{{\SetFigFont{7}{8.4}{\familydefault}{\mddefault}{\updefault}{\color[rgb]{0,0,0}crash $p_1$}%
}}}}
\put(5626,-6361){\makebox(0,0)[lb]{\smash{{\SetFigFont{7}{8.4}{\familydefault}{\mddefault}{\updefault}{\color[rgb]{0,0,0}crash $p_1$}%
}}}}
\put(5626,-7186){\makebox(0,0)[lb]{\smash{{\SetFigFont{7}{8.4}{\familydefault}{\mddefault}{\updefault}{\color[rgb]{0,0,0}$\phi$}%
}}}}
\put(3751,-3511){\makebox(0,0)[lb]{\smash{{\SetFigFont{7}{8.4}{\familydefault}{\mddefault}{\updefault}{\color[rgb]{0,0,0}(b)}%
}}}}
\put(901,-7861){\makebox(0,0)[lb]{\smash{{\SetFigFont{7}{8.4}{\familydefault}{\mddefault}{\updefault}{\color[rgb]{0,0,0}(c)}%
}}}}
\put(1576,-8236){\makebox(0,0)[lb]{\smash{{\SetFigFont{7}{8.4}{\familydefault}{\mddefault}{\updefault}{\color[rgb]{0,0,0}$op_{i_1}$}%
}}}}
\put(2626,-8236){\makebox(0,0)[lb]{\smash{{\SetFigFont{7}{8.4}{\familydefault}{\mddefault}{\updefault}{\color[rgb]{0,0,0}$op_{j_1}$}%
}}}}
\put(2176,-8011){\makebox(0,0)[lb]{\smash{{\SetFigFont{7}{8.4}{\familydefault}{\mddefault}{\updefault}{\color[rgb]{0,0,0}$q_0$}%
}}}}
\put(2776,-8911){\makebox(0,0)[lb]{\smash{{\SetFigFont{7}{8.4}{\familydefault}{\mddefault}{\updefault}{\color[rgb]{0,0,0}$op_{j_2}$}%
}}}}
\put(1426,-8911){\makebox(0,0)[lb]{\smash{{\SetFigFont{7}{8.4}{\familydefault}{\mddefault}{\updefault}{\color[rgb]{0,0,0}$op_{i_2}$}%
}}}}
\put(2776,-9586){\makebox(0,0)[lb]{\smash{{\SetFigFont{7}{8.4}{\familydefault}{\mddefault}{\updefault}{\color[rgb]{0,0,0}$op_{j_\beta}$}%
}}}}
\put(1426,-9586){\makebox(0,0)[lb]{\smash{{\SetFigFont{7}{8.4}{\familydefault}{\mddefault}{\updefault}{\color[rgb]{0,0,0}$op_{i_\alpha}$}%
}}}}
\put(1426,-10261){\makebox(0,0)[lb]{\smash{{\SetFigFont{7}{8.4}{\familydefault}{\mddefault}{\updefault}{\color[rgb]{0,0,0}$op_{n}$}%
}}}}
\put(2776,-10261){\makebox(0,0)[lb]{\smash{{\SetFigFont{7}{8.4}{\familydefault}{\mddefault}{\updefault}{\color[rgb]{0,0,0}$op_n$}%
}}}}
\put(1726,-10636){\makebox(0,0)[lb]{\smash{{\SetFigFont{7}{8.4}{\familydefault}{\mddefault}{\updefault}{\color[rgb]{0,0,0}$q'$}%
}}}}
\put(1726,-9961){\makebox(0,0)[lb]{\smash{{\SetFigFont{7}{8.4}{\familydefault}{\mddefault}{\updefault}{\color[rgb]{0,0,0}$q$}%
}}}}
\put(2626,-9961){\makebox(0,0)[lb]{\smash{{\SetFigFont{7}{8.4}{\familydefault}{\mddefault}{\updefault}{\color[rgb]{0,0,0}$q$}%
}}}}
\put(2626,-10636){\makebox(0,0)[lb]{\smash{{\SetFigFont{7}{8.4}{\familydefault}{\mddefault}{\updefault}{\color[rgb]{0,0,0}$q'$}%
}}}}
\put(3751,-7936){\makebox(0,0)[lb]{\smash{{\SetFigFont{7}{8.4}{\familydefault}{\mddefault}{\updefault}{\color[rgb]{0,0,0}(d)}%
}}}}
\put(4276,-8986){\makebox(0,0)[lb]{\smash{{\SetFigFont{7}{8.4}{\familydefault}{\mddefault}{\updefault}{\color[rgb]{0,0,0}$op_{i_2}$}%
}}}}
\put(4426,-8236){\makebox(0,0)[lb]{\smash{{\SetFigFont{7}{8.4}{\familydefault}{\mddefault}{\updefault}{\color[rgb]{0,0,0}$op_{i_1}$}%
}}}}
\put(5476,-8236){\makebox(0,0)[lb]{\smash{{\SetFigFont{7}{8.4}{\familydefault}{\mddefault}{\updefault}{\color[rgb]{0,0,0}$op_n$}%
}}}}
\put(5026,-8011){\makebox(0,0)[lb]{\smash{{\SetFigFont{7}{8.4}{\familydefault}{\mddefault}{\updefault}{\color[rgb]{0,0,0}$q_0$}%
}}}}
\put(4276,-9586){\makebox(0,0)[lb]{\smash{{\SetFigFont{7}{8.4}{\familydefault}{\mddefault}{\updefault}{\color[rgb]{0,0,0}$op_{i_\alpha}$}%
}}}}
\put(4576,-9961){\makebox(0,0)[lb]{\smash{{\SetFigFont{7}{8.4}{\familydefault}{\mddefault}{\updefault}{\color[rgb]{0,0,0}$q_0$}%
}}}}
\put(4276,-10261){\makebox(0,0)[lb]{\smash{{\SetFigFont{7}{8.4}{\familydefault}{\mddefault}{\updefault}{\color[rgb]{0,0,0}$op_{n}$}%
}}}}
\put(4276,-12211){\makebox(0,0)[lb]{\smash{{\SetFigFont{7}{8.4}{\familydefault}{\mddefault}{\updefault}{\color[rgb]{0,0,0}$op_{i_2}$}%
}}}}
\put(5626,-12211){\makebox(0,0)[lb]{\smash{{\SetFigFont{7}{8.4}{\familydefault}{\mddefault}{\updefault}{\color[rgb]{0,0,0}$op_{j_2}$}%
}}}}
\put(4426,-11461){\makebox(0,0)[lb]{\smash{{\SetFigFont{7}{8.4}{\familydefault}{\mddefault}{\updefault}{\color[rgb]{0,0,0}$op_{i_1}$}%
}}}}
\put(5476,-11461){\makebox(0,0)[lb]{\smash{{\SetFigFont{7}{8.4}{\familydefault}{\mddefault}{\updefault}{\color[rgb]{0,0,0}$op_{j_1}$}%
}}}}
\put(5026,-11236){\makebox(0,0)[lb]{\smash{{\SetFigFont{7}{8.4}{\familydefault}{\mddefault}{\updefault}{\color[rgb]{0,0,0}$q_0$}%
}}}}
\put(3751,-11161){\makebox(0,0)[lb]{\smash{{\SetFigFont{7}{8.4}{\familydefault}{\mddefault}{\updefault}{\color[rgb]{0,0,0}(f)}%
}}}}
\put(4276,-12811){\makebox(0,0)[lb]{\smash{{\SetFigFont{7}{8.4}{\familydefault}{\mddefault}{\updefault}{\color[rgb]{0,0,0}$op_{i_\alpha}$}%
}}}}
\put(5626,-12811){\makebox(0,0)[lb]{\smash{{\SetFigFont{7}{8.4}{\familydefault}{\mddefault}{\updefault}{\color[rgb]{0,0,0}$op_{j_\beta}$}%
}}}}
\put(5476,-13186){\makebox(0,0)[lb]{\smash{{\SetFigFont{7}{8.4}{\familydefault}{\mddefault}{\updefault}{\color[rgb]{0,0,0}$q_0$}%
}}}}
\put(4576,-13186){\makebox(0,0)[lb]{\smash{{\SetFigFont{7}{8.4}{\familydefault}{\mddefault}{\updefault}{\color[rgb]{0,0,0}$q_0$}%
}}}}
\put(2776,-5911){\makebox(0,0)[lb]{\smash{{\SetFigFont{7}{8.4}{\familydefault}{\mddefault}{\updefault}{\color[rgb]{0,0,0}$\delta$}%
}}}}
\put(901,-4711){\makebox(0,0)[lb]{\smash{{\SetFigFont{7}{8.4}{\familydefault}{\mddefault}{\updefault}{\color[rgb]{0,0,0}$v_2$-valent}%
}}}}
\put(2926,-4711){\makebox(0,0)[lb]{\smash{{\SetFigFont{7}{8.4}{\familydefault}{\mddefault}{\updefault}{\color[rgb]{0,0,0}$v'$-valent}%
}}}}
\put(901,-3886){\makebox(0,0)[lb]{\smash{{\SetFigFont{7}{8.4}{\familydefault}{\mddefault}{\updefault}{\color[rgb]{0,0,0}$v_2\neq v'$}%
}}}}
\put(901,-3511){\makebox(0,0)[lb]{\smash{{\SetFigFont{7}{8.4}{\familydefault}{\mddefault}{\updefault}{\color[rgb]{0,0,0}(a)}%
}}}}
\put(2776,-5011){\makebox(0,0)[lb]{\smash{{\SetFigFont{7}{8.4}{\familydefault}{\mddefault}{\updefault}{\color[rgb]{0,0,0}$p_2$}%
}}}}
\put(2326,-3661){\makebox(0,0)[lb]{\smash{{\SetFigFont{7}{8.4}{\familydefault}{\mddefault}{\updefault}{\color[rgb]{0,0,0}$\gamma$}%
}}}}
\put(2551,-4336){\makebox(0,0)[lb]{\smash{{\SetFigFont{7}{8.4}{\familydefault}{\mddefault}{\updefault}{\color[rgb]{0,0,0}crash $p_1$}%
}}}}
\put(1651,-4336){\makebox(0,0)[lb]{\smash{{\SetFigFont{7}{8.4}{\familydefault}{\mddefault}{\updefault}{\color[rgb]{0,0,0}$p_2$}%
}}}}
\put(1426,-12211){\makebox(0,0)[lb]{\smash{{\SetFigFont{7}{8.4}{\familydefault}{\mddefault}{\updefault}{\color[rgb]{0,0,0}$op_{i_2}$}%
}}}}
\put(1576,-11461){\makebox(0,0)[lb]{\smash{{\SetFigFont{7}{8.4}{\familydefault}{\mddefault}{\updefault}{\color[rgb]{0,0,0}$op_{i_1}$}%
}}}}
\put(2626,-11461){\makebox(0,0)[lb]{\smash{{\SetFigFont{7}{8.4}{\familydefault}{\mddefault}{\updefault}{\color[rgb]{0,0,0}$op_k$}%
}}}}
\put(2176,-11236){\makebox(0,0)[lb]{\smash{{\SetFigFont{7}{8.4}{\familydefault}{\mddefault}{\updefault}{\color[rgb]{0,0,0}$q_0$}%
}}}}
\put(1426,-12811){\makebox(0,0)[lb]{\smash{{\SetFigFont{7}{8.4}{\familydefault}{\mddefault}{\updefault}{\color[rgb]{0,0,0}$op_{i_\alpha}$}%
}}}}
\put(1726,-13186){\makebox(0,0)[lb]{\smash{{\SetFigFont{7}{8.4}{\familydefault}{\mddefault}{\updefault}{\color[rgb]{0,0,0}$q_0$}%
}}}}
\put(1426,-13486){\makebox(0,0)[lb]{\smash{{\SetFigFont{7}{8.4}{\familydefault}{\mddefault}{\updefault}{\color[rgb]{0,0,0}$op_{k}$}%
}}}}
\put(901,-5086){\makebox(0,0)[lb]{\smash{{\SetFigFont{7}{8.4}{\familydefault}{\mddefault}{\updefault}{\color[rgb]{0,0,0}$\delta$: $p_2$ takes}%
}}}}
\end{picture}%

%% file: readable-consequences.tex

\subsection{Relationship Between Consensus and Recoverable Consensus}
\label{sec:relationship}

Next, we prove a relationship between the characterizations of  types that  solve
consensus and those that  solve~RC.

\begin{theorem}
\label{discerning}
For $n\geq 4$, if a deterministic type $T$ is $n$-discerning, then it is \cond{$(n-2)$}.
\end{theorem}
\begin{proof}
Let $q_0, A, B, op_1, \ldots, op_n$ be chosen to satisfy the definition of $n$-discerning.
Without loss of generality, assume that $\{p_1, \ldots, p_{n-2}\}$ includes at least one process
from each of  $A$ and $B$, and that $\{p_{n-1}, p_n\}$ includes at least one process from 
each team that contains more than one process.
(The ids of the processes can be permuted to make this true.)
We partition the processes $\{p_1, \ldots, p_{n-2}\}$ into two non-empty teams
$A' = A \cap \{p_1, \ldots, p_{n-2}\}$ and
$B' = B \cap \{p_1, \ldots, p_{n-2}\}$.

We show that $Q_{A'}(q_0,op_1,\ldots,op_{n-2})$ and $Q_{B'}(q_0,op_1,\ldots,op_{n-2})$ satisfy the definition of \cond{$(n-2)$}.

To derive a contradiction, assume $Q_{A'} \cap Q_{B'}$ contains some state~$q$.
Then, there are sequences $i_1,\ldots,i_\alpha$ and $j_1,\ldots,j_\beta$, each of distinct ids  from $\{1,\ldots,n-2\}$, such that $p_{i_1}\in A$, $p_{j_1}\in B$ and the sequences
$op_{i_1},\ldots,op_{i_\alpha}$ and $op_{j_1},\ldots,op_{j_\beta}$ both take an object of type $T$ from
state $q_0$ to  $q$.
Operation $op_n$ takes the object from state $q$ to some state $q'$ and returns some response $r$.
By adding $op_n$ to the end of each of the two sequences,
we see the pair $(r,q')$ is in both $R_{A,n}$ and $R_{B,n}$
in the definition of $n$-discerning, a contradiction.
Thus, condition \ref{empty-cond} of the definition of \cond{$(n-2)$} is satisfied.

To derive a contradiction, assume $q_0\in Q_{A'}$ and $|B'|>1$.
Since $|B| \geq |B'| > 1$, some process $p_j$ is in $B \cap \{p_{n-1},p_n\}$.
Operation $op_j$ takes an object of type $T$ from $q_0$ to some state $q$ and returns some response $r$.
Thus, $(r,q)$ is in the set $R_{B,j}$ of the definition of $n$-discerning.
Since $q_0\in Q_{A'}$, there is a sequence $i_1,\ldots,i_\alpha$ of distinct ids chosen from $\{1,\ldots,n-2\}$ such that $p_{i_1}\in A$ and the sequence
$op_{i_1},\ldots,op_{i_\alpha}$ takes an object of type $T$ from
state $q_0$ back to the state~$q_0$.
By adding $op_j$ to the end of this sequence, we see that the pair $(r,q)$ is also in $R_{A,j}$,
contradicting the fact that $R_{A,j}\cap R_{B,j}$ must be empty, according to the definition of $n$-discerning.
Thus, condition \ref{init-A-cond} of the definition of \cond{$(n-2)$} is satisfied.

The proof of condition \ref{init-B-cond} 
is symmetric.
\end{proof}

\begin{corollary}
\label{drop-by-2}
A deterministic, readable object type $T$ with consensus number \x{at least} $n$
can solve recoverable consensus among $n-2$ processes.
Thus, $cons(T) -2 \leq rcons(T) \leq cons(T)$.   
\end{corollary}

The first inequality in the  corollary is a consequence of Theorem \ref{sufficient} and \ref{discerning}.  The second inequality follows from the fact that any algorithm that solves
RC is also an algorithm that solves consensus.

For  $n=3$, we can strengthen Theorem \ref{discerning} and Corollary \ref{drop-by-2} as follows.  See \arxcam{Appendix \ref{3dis2rec-proof}}{the full version \cite{full}} for the proof.
\begin{proposition}
\label{3dis2rec}
If a deterministic, readable type is 3-discerning, then it is \cond{2}.
Thus, if $cons(T)=3$ then $2\leq rcons(T)\leq 3$.
\end{proposition}

The following example shows that Theorem \ref{discerning} cannot be strengthened when $n>3$.

\newcommand{\tA}{\mathbb{A}}
\newcommand{\tB}{\mathbb{B}}

\begin{proposition}
\label{discerning-ceg}
For all $n\geq 4$, there is a type that is $n$-discerning, but is not \cond{$(n-1)$}.
\end{proposition}

A complete proof is in \arxcam{Appendix \ref{Tn-proof}}{\cite{full}}.  We sketch it here.
We define a type $T_n$ whose set of  states is 
$\{(winner,row,col): winner \in \{\tA,\tB\}, 0 \leq row < \ceil{n/2},  0 \leq col < \floor{n/2}\} \cup \{(\bot,0,0)\}$.
$T_n$ has two operations $op_\tA$ and $op_\tB$, and a read operation.
Intuitively, if the object is initialized to $(\bot,0,0)$, $winner$ keeps track of whether
the first update was $op_\tA$ or $op_\tB$, while $col$ and $row$ store the number of times
$op_\tA$ and $op_\tB$ have been applied.  If $op_\tA$ is performed more than $\floor{n/2}$
times or $op_\tB$ is performed more than $\ceil{n/2}$ times,  the object ``forgets'' all the
information it has stored by going back to state $(\bot,0,0)$.
It is easy to verify that $T_n$ is $n$-discerning but not \cond{$(n-1)$}.

It follows easily from Proposition \ref{discerning-ceg} combined with Theorems \ref{cons-char} and \ref{necessary} that there are readable types whose RC numbers are strictly smaller than their consensus numbers.
\begin{corollary}
For all $n\geq 4$, there is a deterministic, readable type $T_n$ such that $rcons(T_n) < cons(T_n) =n$.
\end{corollary} 

On the other hand, there are also  types whose RC numbers are equal to their consensus numbers.  The next proposition also shows that every level of the RC hierarchy is populated, since there are types with consensus number $n$ for all $n$.

\begin{proposition}
\label{same-number}
For all $n$, there is a deterministic, readable type $S_n$ such that $rcons(S_n) = cons(S_n) = n$.
\end{proposition}

A complete proof is in \arxcam{Appendix \ref{Sn-proof}}{the full version \cite{full}}.  We sketch it here.
We define a type $S_n$ whose set of possible states is
$\{(winner,row) : winner \in \{\tA,\tB\}, 0\leq row < n\}$.
$S_n$ has two operations $op_\tA$ and $op_\tB$, and a read operation.
Intuitively, if the object is initialized to $(\tB,0)$, and then accessed by update operations,
$winner$ records whether the first 
update was $op_\tA$ or $op_\tB$ and $row$ counts the number of times $op_\tB$ has been applied.  If $op_\tA$ is performed more than once or if $op_\tB$ is performed more than $n-1$ times, then the object
``forgets'' all the information it has stored by going back to state $(\tB,0)$.
It is fairly straightforward to check that $S_n$ is \cond{$n$}, but is
not $(n+1)$-discerning.  Thus, $n\leq rcons(S_n) \leq cons(S_n) \leq n$.

\subsection{Recoverable Consensus Using Several Types}
\label{sec:robustness}

The (recoverable) consensus number
of a set ${\mathcal T}$ of object types is the maximum number of processes that can
solve (recoverable) consensus using objects of those types, together with registers
(or $\infty$ if there is no such maximum).
A classic open question, originally formulated by Jayanti \cite{Jay97}, is whether the standard consensus hierarchy is robust
for deterministic types,
i.e., whether $cons({\mathcal T})=\max\{cons(T) : T\in{\mathcal T}\}$.
If \arxcam{this equation holds}{so}, it is possible to study the power of a system
equipped with multiple types by studying the power of each type individually.
See \cite[Section~9]{FR03} for some history of the robustness question.
Ruppert's characterization  (Theorem \ref{cons-char})
was used to show the consensus hierarchy is robust for the class of deterministic,
\emph{readable} types.
Similarly, our  characterization allows us to show how the power of a set of
deterministic, readable types to solve RC is related to the power of the
individual types.

\begin{theorem}
\label{robustness}
Let ${\mathcal T}$ be a non-empty set of deterministic, readable types and suppose
$n=\max\{rcons(T) : T\in{\mathcal T}\}$ exists.  Then,
$n \leq rcons({\mathcal T})\leq n+1$.
(If $\max\{rcons(T) : T\in{\mathcal T}\}$ does not exist, then $rcons({\mathcal T})=\infty$.)
\end{theorem}
\begin{proof}
If $\max\{rcons(T) : T\in{\mathcal T}\}$ does not exist, then for any $n$
there is an algorithm that solves RC using some type $T_n\in {\mathcal T}$.
It follows that $rcons({\mathcal T})=\infty$.  So for the remainder of the proof,
assume the maximum does exist.

It follows from the definition that $rcons(\mathcal{T}) \geq rcons(T)$ for all 
$T\in{\mathcal T}$.  Thus, $rcons(\mathcal{T})\geq \max\{rcons(T) : T\in{\mathcal T}\}$.

We prove the other inequality by contradiction.
Suppose $(n+2)$-process RC can be solved
using types in ${\mathcal T}$.
As in the proof of Theorem \ref{necessary}, there is a critical execution $\gamma$
at the end of which each process is about to update the same object $O$ of some type 
$T\in {\mathcal T}$.  As in that proof,
$T$ is \cond{$(n+1)$}.
By Theorem \ref{sufficient}, there is an $(n+1)$-process RC algorithm
using objects of type $T$ and registers.  So, 
$rcons(T) \geq n+1 > n \geq rcons(T)$, a contradiction.
\end{proof}

%% file: universality.tex

\section{The Significance of Recoverable Consensus}
\label{universal}

Herlihy's universal construction \cite{Her91} builds a linearizable, wait-free implementation of \emph{any} shared object using a consensus algorithm as a subroutine.
It creates a linked list of all operations performed on the implemented object, and this list defines the linearization ordering.
Berryhill, Golab and Tripunitara \cite{BGT15}  observed that this result extends
to the model with simultaneous crashes, simply by placing the linked list in non-volatile memory
and using RC in place of consensus.
Their model allows a part of shared memory to be volatile.
Using that volatile memory, their universal construction
provides strictly linearizable implementations.
(\emph{Strict linearizability} \cite{AF03} is similar to linearizability, with the requirement that
an operation in progress when a process crashes is either linearized before the crash or not at all.)
Without volatile shared memory, the history satisfies only the weaker property of \emph{recoverable linearizability} (proposed in \cite{BGT15}, with a correction to the definition in~\cite{Li21}).

Similarly, we observe that Herlihy's universal construction 
also extends to the independent crash model.
To execute an operation $op$, a process
creates a node $nd$ containing $op$ (including its parameters). 
Then, it announces $op$ by storing a pointer to $nd$ in an announcement array. 
Other processes can then help add $op$ to the list, ensuring wait-freedom. 
Processes use an instance of consensus  to agree on the next pointer of each node
in the list.
A process executes a routine \Perform\ that traverses the list.
At each visited node, it proposes a value from the announcement array to the 
consensus algorithm  for the node's next pointer,
until it discovers its own operation's node $nd$ has been appended.
Processes choose which announced value to propose so that each process's
announced value is given priority in a round-robin fashion.  This ensures each announced
node is appended within a finite number of steps.

In our setting, all shared variables are non-volatile, and we use an algorithm for RC
(such as the one in Section \ref{sec:sufficient}) in place of consensus.
For simplicity, we use a standard assumption
(as in, e.g.,~\cite{ABH18,ABF+22,FKK22,CFR18,RamalheteCFC19,CorreiaFR20,FB+20,FPR21}):
when a process recovers from a crash, it executes a {\em recovery function}. 
This assumption is not restrictive; 
we could, alternatively, add the code of the recovery function 
at the beginning of the universal algorithm, thus
forcing every process to execute this code before it actually
starts executing a new operation.
When a process $p$ crashes and recovers, the recovery function checks if the last operation 
that $p$  announced before crashing has been appended in the list and if not, it executes the code to append it. 
Specifically, the recovery function simply calls \Perform\ for the last announced node of $p$.
See \arxcam{Appendix \ref{universal-app}}{the full version \cite{full}} for pseudocode of the recoverable universal construction \RUniversal.

As in Herlihy's construction, the helping mechanism of \RUniversal\ ensures wait-freedom.
The recoverable implementations obtained using \RUniversal\
satisfy 
{\em nesting-safe recoverable linearizability} (NRL)~\cite{ABH18}, which requires that 
a crashed operation is linearized within an interval that includes its crashes and recovery attempts. 
NRL implies \emph{detectability}~\cite{ABH18} which ensures that 
a process can discover upon recovery whether or not its last operation took effect, 
and guarantees that if it did, its response value was made persistent. 
Other well-known safety conditions for the crash-recovery setting include
{\em durable linearizability}~\cite{IMS16},  
which has been proposed for the system-crash failures model and requires that the effects of all operations
that have completed before a crash are reflected in the object's state upon recovery,
and {\em persistent linearizability}~\cite{GL04}, which has been proposed for a model
where no recovery function is provided and requires that an operation interrupted 
by a crash can be linearized up until the invocation of the next operation by the same process. 
With minor adjustments these conditions are meaningful in our setting
and \RUniversal\ satisfies both of them. 

Moreover, \RUniversal\ has the following desirable property. 
Suppose an implementation $I$ uses a linearizable object~$X$ in a system
with halting failures, but no crash-recovery failures.
We can transform $I$ to an implementation $I'$ by replacing every instance of $X$ in $I$ 
with an invocation of \RUniversal\ (that implements $X$). 
Then, every trace produced by $I'$ in a system with crash and recovery failures
is also a trace of 
$I$ using a linearizable object $X$ in a system with halting failures.
In this way, any algorithm designed for the standard asynchronous model with halting failures 
can be automatically transformed to another algorithm to run in the independent crash-recovery model, 
as long as we can solve RC.  

The traditional consensus hierarchy gives us information about which implementations
are possible (via universality), but also
tells us some implementations
are impossible. This is another reason to study the consensus hierarchy.
Specifically, if $cons(T_1)<cons(T_2)$, then
there is no wait-free implementation of object type $T_2$ from objects of type $T_1$ 
for more than $cons(T_1)$ processes \cite{Her91}.
We give an analogous result for the RC hierarchy. 
\arxcam{The proof is in Appendix~\ref{rcons-implement-proof}.}{For the proof, see \cite{full}.}

\begin{theorem}
\label{rcons-implement}
Let $n\leq rcons(T_2)$.
If there is a wait-free, persistently linearizable implementation of $T_2$ from atomic 
objects of type $T_1$ (and registers) in a system of $n$ processes with independent crashes, 
then $rcons(T_1)\geq n$.
\end{theorem}


\begin{corollary}
If $rcons(T_1)< rcons(T_2)$ then there is no wait-free, persistently linearizable implementation
of $T_2$ from atomic objects of type $T_1$ and registers in a system of more than
$rcons(T_1)$ processes with independent crashes.
\end{corollary}

%% file: conclusion.tex

\section{Discussion}

In this paper, we studied solvability, without considering
efficiency. A lot of research has focused on designing efficient recoverable
transactional memory systems~\cite{VT+11,CC+11-I,CD+14-I,RamalheteCFC19,CGZ18,BC+20}) 
and recoverable universal constructions~\cite{CorreiaFR20,FKK22}.
Wait-free solutions appear in~\cite{RamalheteCFC19,CorreiaFR20,FKK22}. 
Some~\cite{CorreiaFR20,FKK22} are based on existing wait-free universal constructions~\cite{CRP20,FK14}
for the standard shared-memory model with halting failures.
All except~\cite{FKK22}, 
satisfy weaker consistency conditions than nesting-safe recoverable linearizability.
Attiya {\em et al.}~\cite{ABH18} gave  a recoverable implementation of a Compare\&Swap (CAS)
object. Any concurrent algorithm from read/write and CAS objects can become recoverable 
by replacing its CAS objects with their recoverable implementation~\cite{ABH18}. 
Capsules~\cite{NormOptQueue19} can also be used to 
transform concurrent algorithms that use only read and CAS primitives
to their recoverable versions.
Many other general techniques~\cite{ABF+22,FPR21,FB+20} have been proposed for deriving recoverable 
lock-free data structures from their concurrent implementations.

Our work leaves open several questions. 
Is there a deterministic, readable type $T$ with $rcons(T)=cons(T)-2$?
We saw in Corollary \ref{drop-by-2} that $cons(T)-rcons(T)$ can be at most 2 for deterministic, readable types.
How big can this difference be for non-readable~types?

It would be nice to close the gap between the necessary condition of being \cond{$(n-1)$}
and the sufficient condition of being \cond{$n$} for the solvability of 
RC using deterministic, readable types.
Perhaps a good starting point  is to determine whether being
\cond{2} is actually necessary for solving 2-process RC. 
\ignore{
\here{Can omit or compact rest of this parag if we need space}
In this case, $|A|=|B|=1$, so only the condition $Q_A\cap Q_B$ must be checked.
Let $q_A$ and $q_B$ be the states after $op_A$ or $op_B$ is performed on an object in state $q_0$.
Let $q_{AB}$ and $q_{BA}$ be the states after both operations are performed in either order.
Then, the \cond{2} property states that $\{q_A,q_{AB}\} \cap \{q_B,q_{BA}\}$ is empty.
One of the main technical results of Golab \cite{Gol20} showed that if
2-process RC is solvable, then there
are $q_0, op_1, op_2$ such that the operations neither commute nor overwrite.
This can be restated
as saying $q_{BA} \notin \{q_A,q_{AB}\}$ and $q_{AB}\notin \{q_B,q_{BA}\}$.
So it just remains to check whether 2-process consensus can be solved if $q_A=q_B$.
}
Finally, it would be interesting to  characterize 
read-modify-write types capable of solving $n$-process RC (as was done 
in \cite{Rup00} for the standard consensus problem), and see whether 
the RC hierarchy is robust for deterministic, readable types 
(or for all deterministic types).

%% file: appendix.tex

\pagebreak

\appendix

\section{Proof of Theorem \ref{system-thm}}
\label{system-proof}

\newcommand{\decide}{\mbox{\sc Decide}}
\newcommand{\pref}{\mbox{\it pref}}
\newcommand{\rrun}{run} 

\begin{figure}[b]
{\small
\begin{code}
\firstline
shared variables:\nl
\n  array $Round[1..n]$ of registers, initially 0\nl
    array $D[1..\infty]$ of registers, initially $\bot$\bl
\p \nl
\decide($v$)\nl
\n  $\pref \leftarrow v$ \llabel{init-pref}\nl
	$r\leftarrow 1$\nl
    loop\nl
\n     if $Round[j] < r$ then \llabel{read-round}\nl
\n         $Round[j] \leftarrow r$ \llabel{write-round}\nl
           if $r>1$ and $D[r-1] \neq \bot$ then \nl
\n            $\pref \leftarrow D[r-1]$ \llabel{update-pref1}\nl
\p         end if\nl
           $\pref \leftarrow C_r.\decide(\pref\,)$ \llabel{call-cons}\nl
           $D[r] \leftarrow \pref$ \llabel{write-D}\nl
           if $\forall k$, $Round[k] \leq r$ then \llabel{collect-round}\nl
\n              return \pref \llabel{output}\nl
\p         end if\nl
\p     else if $r>1$ and $D[r-1]\neq \bot$ then\nl
\n         $\pref\leftarrow D[r-1]$ \llabel{update-pref2}\nl
\p     end if\nl
       $r \leftarrow r+1$\nl
\p  end loop\nl
\p end \decide
\end{code}
}
\caption{\label{system-alg}Algorithm for process $p_j$ to solve recoverable consensus with simultaneous crashes using instances of standard consensus $C_1, C_2, \ldots$.}
\end{figure}

We must show that RC is solvable for $n$ processes with simultaneous crashes if and only if the standard consensus problem is solvable for $n$ processes.
The ``only if'' direction is trivial, since any algorithm for RC also solves
consensus:  just consider executions in which there are no crashes.

We prove the converse using the algorithm  shown in Figure~\ref{system-alg}, which is similar to the algorithm in \cite{Gol20} for a bounded number of crashes.
The RC algorithm uses multiple instances of the consensus algorithm, denoted $C_1, C_2, \ldots$.
Each process attempts to access $C_1, C_2, \ldots$ in turn, until it receives a result from one of them and sees that no process has yet moved on to the next object.
The register $D[r]$ is used to record the output of $C_r$.
An iteration of the loop with $r = i$ is called a \emph{\rrun\ of round $i$}.
A process may \rrun\ round $i$ multiple times if it crashes and recovers.
Process $p_j$ records in register $Round[j]$ the largest $r$ for which $p_j$ has started to \rrun\ round $i$.
This variable is used to ensure that $p_j$ does not access $C_r$ a second time if
it crashes during a \rrun\ of round $r$.

\begin{lemma}
The algorithm in Figure \ref{system-alg} satisfies the recoverable wait-freedom property.
\end{lemma}
\begin{proof}
To derive a contradiction, suppose there is a time $t$ after which there are no more crashes and processes continue to take steps without any process terminating.
Let $i$ be greater than the maximum entry in $Round[1..n]$ at time $t$.
Any process that executes enough steps after $t$ will either terminate, crash or reach 
round $i$.  The processes that reach round $i$ without crashing will satisfy the test at line \lref{read-round}.  Among those processes, the first to complete line \lref{collect-round} will satisfy the test and terminate, a contradiction.
\end{proof}

\begin{observation}
\label{round-incr}
For any $j$, the value in $Round[j]$ only increases.
\end{observation}
\begin{proof}
Only process $p_j$ writes to $Round[j]$.  Moreover, $p_j$ writes a 
value $r$ to $Round[j]$ (at line \lref{write-round})
only after seeing (at line \lref{read-round}) that the current value of $Round[j]$ is less than $r$.
\end{proof}

\begin{lemma}
\label{C-safe}
For each $i$, no process invokes \decide\ on $C_i$ more than once.
\end{lemma}
\begin{proof}
Before $p_j$ invokes \decide\ on $C_i$ for the first time, $p_j$ writes $i$ to $Round[j]$.
By Observation \ref{round-incr}, $Round[j] \geq i$ at all times after that invocation.
Thus, any subsequent test by $p_j$ at line \lref{read-round} of a \rrun\ of round $i$ will fail, 
and $p_j$ will never invoke \decide\ on $C_i$ again.
\end{proof}

Since $C_i$ is accessed correctly, even if processes crash and recover, it follows that
calls to $C_i.\decide$ cannot return different values and that 
the common output value is one of the input values to one of the calls to $C_i.\decide$.
In particular, this means that all values written to $D[i]$ at line \lref{write-D} are identical.

\begin{lemma}
\label{valid}
Let $P_0$ be the set of input values to the RC algorithm.
For $r\geq 1$, let $P_r$ be the set of all values that are either written into $D[r]$, or stored 
in the local variable \pref\ of some process 
when it completes a \rrun\ of round $r$ by terminating at line \lref{output} or reaching 
the end of the iteration of the loop.  
For $r\geq 1$, $P_r \subseteq P_{r-1}$.
\end{lemma}
\begin{proof}
Let $P_r'$ be the set of all values that are in the local variable \pref\ of some process at some time 
during a \rrun\ of round $r$ by that process.  Since $P_r \subseteq P_r'$, we must simply
show that $P_r' \subseteq P_{r-1}$.
Let $v\in P_r'$.  
Setting \pref\  at line \lref{call-cons} does not add any new value to $P_r'$, by the validity 
property of the consensus algorithm $C_r$.
Thus, some process's \pref\ was either equal to $v$ at the beginning of a \rrun\ of round~$r$, 
or set to $v$ at line \lref{update-pref1} or \lref{update-pref2}.  We show that in each of these cases, $v\in P_{r-1}$.
If $r=1$ and \pref\ is $v$ at the start of a \rrun\ of round 1, then \pref\ was set to $v$ at line \lref{init-pref}, so $v\in P_0$.
If $r>1$ and \pref\ is $v$ at the start of a \rrun\ of round $r$, then \pref\ was $v$ at the end of $p_j$'s \rrun\ of round $r-1$, so $v\in P_{r-1}$.
If \pref\ is set to $v$ at line \lref{update-pref1} or \lref{update-pref2}, then $r>1$ and
$v$ was stored in $D[r-1]$, so $v\in P_{r-1}$.
\end{proof}

The validity property follows from Lemma \ref{valid}:  if a process returns $x$ at line \lref{output} in a \rrun\ of round $r$, then $x\in P_r \subseteq P_0$.
Thus, all output values are in the set $P_0$ of input values.

\begin{lemma}
\label{agreement}
The algorithm in Figure \ref{system-alg} satisfies the agreement property.
\end{lemma}
\begin{proof}
Consider an execution in which some processes return a value.
Let $i$ be the minimum number such that some process returns a value at line \lref{output} 
of a \rrun\ of round $i$.  
Any process that returns a value at line \lref{output} during a \rrun\ of round $i$ returns the value produced
by $C_i$ at line \lref{call-cons}.
By Lemma \ref{C-safe} and the agreement property of $C_i$, all such processes return the same value $x$,
and $x$ is the only non-$\bot$ value that can ever be stored in $D[i]$.

The process that returns a value $x$ at line \lref{output} in its \rrun\ of round $i$ first
writes $x$ into $D[i]$ at line \lref{write-D} and then reads values less than or equal to $i$ from each
entry of $Round[1..n]$ at line \lref{collect-round}.
By Observation \ref{round-incr}, no entry of $Round$ is greater than $i$ at any time before $x$ is first written 
into $D[i]$.

We show that $P_{i+1}$ (as defined in Lemma \ref{valid}) can only contain $x$.
Consider a process $p_j$ that either writes to $D[i+1]$ or completes a \rrun\ of round $i+1$.
We consider two cases.
First, suppose $Round[j] \geq r$ at line \lref{read-round} of $p_j$'s \rrun\ of round $i+1$.
By the argument of the previous paragraph, $x$ has already been written into $D[i]$,
so $p_j$ updates its \pref\ to $x$ at line~\lref{update-pref2}.
Otherwise, $Round[j] < r$ at line \lref{read-round} of a \rrun\ of round $i+1$.
In this case, line \lref{write-round}
must be performed after $x$ has been written to $D[i]$, as argued in the previous paragraph.
Thus, $p_j$ updates its \pref\ to $x$ at line \lref{update-pref1}.

Since all processes $p_j$ that execute round $i+1$ update \pref\ to $x$, it follows that all inputs to $C_{i+1}$ are $x$.
By the validity property of $C_{i+1}$, processes can only update their \pref\ to $x$ at line \lref{call-cons}
of a \rrun\ of round $i+1$, and can only write the value $x$ into $D[i+1]$.

Thus, $P_{i+1}\subseteq \{x\}$.  
Consider any round $i'\geq i+1$.
By Lemma \ref{valid}, $P_{i'} \subseteq P_{i+1} \subseteq \{x\}$.
In particular, this means that any value that is returned in any round $i'\geq i+1$ must also be $x$.
\end{proof}

This completes the proof of Theorem \ref{system-thm}.
\new{The algorithm in Figure \ref{system-alg} uses an unbounded number of instances of consensus.
In the full version of \cite{Gol20}, Golab showed that this is indeed necessary for such a construction.}

\section{Using Recoverable Team  Consensus to Solve Recoverable Consensus}
\label{team-tournament}

Suppose a collection of processes is partitioned into two non-empty teams.
The  recoverable team consensus problem is the same as the RC problem,
except with the precondition that all input values for processes on the same team must be the same.
The following proposition, used in the proof of Theorem \ref{sufficient}, can be proved
in the same way as the analogous claim in \cite{Nei95,Rup00} for standard consensus, but we include the proof here for the 
sake of completeness.

\begin{proposition}
If objects of type $T$ and registers can be used to solve recoverable team consensus among $n$ processes,
then objects of type $T$ and registers can be used to solve recoverable consensus among $n$ processes.
\end{proposition}
\begin{proof}
Assume we have a recoverable team consensus for $n$ processes divided into two
non-empty teams $A$ and $B$.
We use induction on $k$ show that $k$ processes can solve RC using objects
of type $T$ and registers.

For $k=1$, this is trivial:  a process can simply return its own input value.

Let $1<k\leq n$.  Assume the claim holds for fewer than $k$ processes.
We construct an algorithm to solve RC for $k$ processes.
Split the $k$ processes into two non-empty teams $A'$ and $B'$ such that $|A'|\leq |A|$ and $|B'|\leq|B|$.
We use two RC algorithms $R_{A'}$ and $R_{B'}$ for $|A'|$ and $|B'|$ processes, respectively.
These algorithms can be built from objects of type $T$ and registers, by the induction hypothesis since
$|A'|$ and $|B'|$ are less than $k$.
Each process first runs the RC algorithm for its team, and uses the output from it
as the input to a recoverable team consensus algorithm $TC$ to produce the final output.
(Note that the $n$-process recoverable team consensus algorithm still works if only $k$
processes uses it; we think of the other $n-k$ processes simply taking no steps.)

The agreement property of $R_{A'}$ and $R_{B'}$ ensure that the precondition of the recoverable team consensus
algorithm is satisfied.
The agreement property of the $k$-process RC algorithm follows from the
agreement property of $TC$.
The recoverable wait-freedom and validity properties follow from the corresponding properties
of $R_{A'}, R_{B'}$ and $TC$.
\end{proof}

\section{Proof of Proposition \ref{3dis2rec}}
\label{3dis2rec-proof}

\begin{proof}
Let $A,B,q_0,op_1,op_2,op_3$ be chosen to satisfy the definition of 3-discerning for the type $T$.
Without loss of generality, assume $A=\{p_1\}$ and $B=\{p_2,p_3\}$.
Let $A'=\{p_1\}$ and $B'=\{p_2\}$.  We show that $A',B',q_0,op_1,op_2$ satisfy the definition
of \cond{2}.
Since $|A'|=|B'|=1$, conditions \ref{init-A-cond} and \ref{init-B-cond} of the definition of
\cond{2} are trivially satisfied.
The argument that condition \ref{empty-cond} is satisfied is similar to the proof of Theorem \ref{discerning}:  if some state $q$ can be reached by two sequences starting with operations on
opposite teams, we can append $op_3$ to the end of each to show that $R_{A,3} \cap R_{B,3} \neq \emptyset$, contradicting the definition of 3-discerning.

Since $T$ is \cond{2}, Theorem \ref{sufficient} implies that $2\leq rcons(T)$.
Moreover $rcons(T)\leq cons(T) = 3$, since any algorithm for RC also solves consensus.
\end{proof}

\section{Proof of Proposition \ref{discerning-ceg}}
\label{Tn-proof}

We use the following definitions of Herlihy \cite{Her91}.
Operations $op_i$ and $op_j$ \emph{commute} from state $q_0$ if 
the sequences $op_i, op_j$ and $op_j, op_i$ take the object from $q_0$ to the same state~$q$.
Operation $op_i$ \emph{overwrites} $op_j$ from $q_0$ if the sequences $op_i$ and $op_j, op_i$ take
the object from $q_0$ to the same state $q$.
It is easy to check that if $q_0, op_1, op_2$ satisfy the definition of \cond{2}, then 
$op_1$ and $op_2$ cannot commute from $q_0$ nor can one overwrite the other from~$q_0$.

\begin{figure*}
\begin{minipage}[t]{.45\textwidth} 
{\small
\begin{code}
\firstline
$op_\tA$\nl
\n  if $winner = \bot$ then\nl
\n      $winner \leftarrow \tA$\nl
        return $\tA$\nl
\p  else \nl
\n      $result \leftarrow winner$\nl
        $col \leftarrow (col + 1) \mbox{ mod }\floor{n/2}$\nl
        if $col = 0$ then\nl
\n          $winner \leftarrow \bot$\nl
            $row = 0$\nl
\p      end if\nl
        return $result$\nl
\p  end if\nl
\p end $op_\tA$
\end{code}
}
\end{minipage}
\hspace*{.1\textwidth}
\begin{minipage}[t]{.45\textwidth} 
{\small
\begin{code}
\firstline
$op_\tB$\nl
\n  if $winner = \bot$ then\nl
\n      $winner \leftarrow \tB$\nl
        return $\tB$\nl
\p  else \nl
\n      $result \leftarrow winner$\nl
        $row \leftarrow (row + 1) \mbox{ mod }\ceil{n/2}$\nl
        if $row = 0$ then\nl
\n          $winner \leftarrow \bot$\nl
            $col = 0$\nl
\p      end if\nl
        return $result$\nl
\p  end if\nl
\p end $op_\tB$
\end{code}
}
\end{minipage}
\vspace*{5mm}

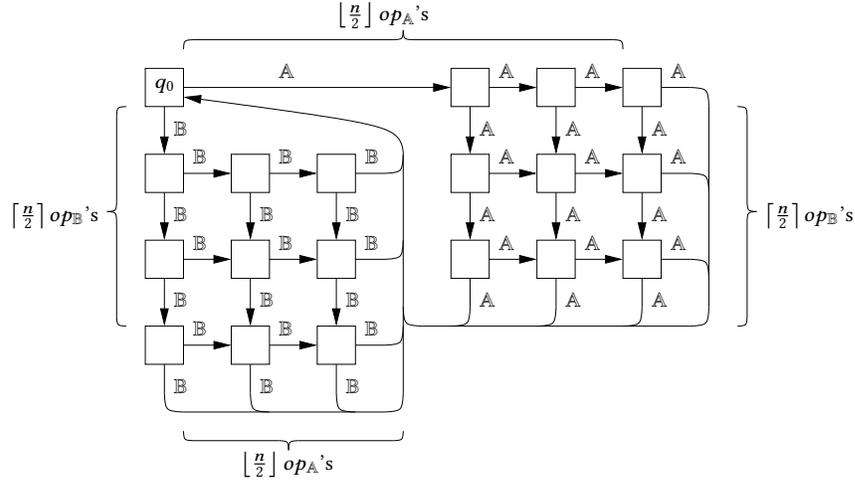
\caption{Behaviour of type $T_n$ used in the proof of Proposition \ref{discerning-ceg}. 
A transition diagram for the object when $n=6$ is shown, where the $op_\tA$ and $op_\tB$ operations are shown
by the horizontal and vertical arrows leaving each state, respectively.  Each transition is labelled by the output value of the operation that causes the transition.
\label{Tn-fig}}
\end{figure*}

\begin{proof}
We define the type $T_n$ as follows.
The set of possible states is $\{(winner,row,col): winner \in \{\tA,\tB\}, 0 \leq row < \ceil{n/2},  0 \leq col < \floor{n/2}\} \cup \{(\bot,0,0)\}$.
$T_n$ supports two operations $op_\tA$ and $op_\tB$, as well as a read operation.
If an update operation is applied to an object in state $(winner,row,col)$, it executes the code in Figure \ref{Tn-fig} atomically to update the state and return a result.

To see that $T_n$ is $n$-discerning,
let $q_0 = (\bot, 0,0)$ and partition processes into team $A$ of size $\floor{n/2}$ and team $B$ of size $\ceil{n/2}$.  Assign $op_\tA$ to all processes on team $A$ and $op_\tB$ to all processes on team $B$.
Then, if any sequence of operations assigned to distinct processes is applied to an object
starting in state $q_0$, every operation's result will be the name of the team of the process
that performed the first operation in the sequence.
This is because the state of the object can return to $q_0$ after a process on 
one team takes the first step
only after {\it all} processes on the other team have taken a step.
Thus, for any $j$, all pairs in $R_{A,j}$ will be of the form $(\tA,*)$ and all pairs in $R_{B,j}$ will be of the form $(\tB,*)$, so $R_{A,j}\cap R_{B,j}=\emptyset$.

It remains to show that $T_n$ is not \cond{$(n-1)$}.
To derive a contradiction, suppose there is some $q_0, A, B, op_1, \ldots,$ $op_{n-1}$ that satisfy the definition of \cond{$(n-1)$}.

If $q_0\neq (\bot,0,0)$, then any pair of operations either commute or overwrite, so even the definition of \cond{2} is not satisfied.
So $q_0$ must be $(\bot,0,0)$.
If two processes $p_i$ and $p_j$ on opposite teams are assigned the same operation, then 
$op_i$ and $op_j$ would both take an object from state $q_0$ to the same state, violating condition
\ref{empty-cond} of the definition of \cond{$(n-1)$}.
Thus, without loss of generality, all processes on team $A$ are assigned $op_\tA$ and all processes
on team $B$ are assigned $op_\tB$.

If $|A|\geq \floor{n/2}$, then allowing one process on team $B$ to take a step followed by $\floor{n/2}$ processes on team $A$ would take the object from state $q_0$ back to $q_0$, so $q_0\in Q_B$.  This
violates condition \ref{init-B-cond} of the definition of \cond{$(n-1)$} since $|A|\geq \floor{4/2}=2$.
Thus, $|A|\leq \floor{n/2}-1$.

Similarly, if $|B|\geq \ceil{n/2}$, then allowing one process on team $A$ to take a step followed by $\ceil{n/2}$ processes on team $B$ would take the object from state $q_0$ back to $q_0$, so $q_0\in Q_A$.  This
violates condition \ref{init-A-cond} of the definition of \cond{$(n-1)$}, since $|B|\geq \ceil{4/2}=2$.
Thus, $|B|\leq \ceil{n/2}-1$.

Hence, $n-1=|A|+|B| \leq (\floor{n/2} -1) + (\ceil{n/2}-1) = n-2$, a contradiction.
\end{proof}

\section{Proof of Proposition \ref{same-number}}
\label{Sn-proof}

\begin{proof}
For $n=1$, let $S_1$ be a type that provides only a read operation.

For $n\geq 2$, we define type $S_n$ as follows.
The set of possible states is $\{(winner,row) : winner \in \{\tA,\tB\}, 0\leq row < n\}$.
$S_n$ supports two operations $op_\tA$ and $op_\tB$, as well as a read operation.
If an update operation is applied to an object in state $(winner, row)$, it executes the code in Figure \ref{Sn-fig} atomically to update the state and return a result.

\begin{figure}
{\small
\begin{code}
\firstline
$op_\tA$\nl
\n  if $(winner,row) = (B,0)$ then\nl
\n      $winner \leftarrow \tA$\nl
\p  else \nl
\n      $winner \leftarrow \tB$\nl
        $row \leftarrow 0$\nl
\p  end if\nl
    return ack\nl
\p end $op_\tA$\bl\nl
$op_\tB$\nl
\n  $row \leftarrow (row+1) \mbox{ mod }n$\nl
    if $row =0$ then\nl
\n      $winner \leftarrow \tB$\nl
\p  end if\nl
    return ack\nl
\p end $op_\tB$
\end{code}
}

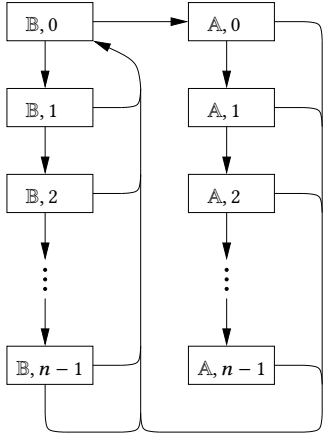
\caption{Behaviour of type $S_n$ used in proof of Proposition \ref{same-number}. 
A transition diagram  is shown, where the $op_\tA$ and $op_\tB$ operations are shown
by the horizontal and vertical arrows leaving each state, respectively.  All operations return ack.
\label{Sn-fig}}
\end{figure}

We first argue that $S_n$ is \cond{$n$}.
Let $q_0=(\tB,0), A=\{p_1\}, B=\{p_2,\ldots,p_n\}, op_1=op_\tA, op_2=op_3=\cdots=op_n=op_\tB$.
Then, $Q_A = \{(\tA,row) : 0\leq row <n\}$ and $Q_B = \{(\tB,row):0\leq row <n\}$ 
satisfy all the conditions of the definition of \cond{$n$}.
It follows from Theorem \ref{sufficient} that $rcons(S_n) \geq n$.

Next, we argue that $S_n$ is not $(n+1)$-discerning, in order to prove that $cons(S_n)\leq n$.
To derive a contradiction, suppose $S_n$ is $(n+1)$-discerning.
Operations assigned to processes on opposite teams cannot commute or overwrite if they
are performed on an object initially in state $q_0$.
(If two operations $op_i$ and $op_j$ assigned to opposite teams commuted, then for
the two sequences $op_i, op_j$ and $op_j, op_i$, the operation $op_i$ gets
the same result ack in both sequences, and both sequences leave the object in the same state.
This would violate the definition of $(n+1)$-discerning.  
A similar argument applies for overwriting operations.)
It is easy to check that $q_0$ must therefore be $(\tB,0)$ and processes on one team (without loss of generality, team $A$) must be assigned the operation $op_\tA$ and processes on the other team $B$ must be assigned the operation $op_\tB$.
If $|A|\geq 2$, then the sequences $op_\tA, op_\tA, op_\tB$ and $op_\tB$ performed on an object would both
take an object from state $(\tB,0)$ to state $(\tB,1)$ and would return the same result for $op_\tB$, violating the definition of $(n+1)$ discerning.
Thus, there must be just one process on team $A$.
So, $|B|=n$.
Consider the following two sequences of operations.
\begin{itemize}
\item
All processes on team $B$ perform $op_\tB$ followed by the process on team $A$ performing $op_\tA$.
\item
The process on team $A$ performs $op_\tA$.
\end{itemize}
Both sequences take the object from state $q_0$ to $(\tA,0)$ and return the same result to $op_\tA$, violating the definition of $(n+1)$-discerning.
This contradiction completes the proof.
\end{proof}

\section{Universal Construction}
\label{universal-app}

The pseudocode for our slightly modified version of Herlihy's universal construction is given in
Figure \ref{universal-code}.
Each list node contains the following fields.
\begin{itemize}
\item $seq$ is initially 0 but is changed to the node's position in the list once the node is added to the list.
\item $op$ is the operation on the implemented object represented by the node; this includes the name of the operation and any arguments to it.
\item $newState$ is the state of the implemented object after the operations on the list up to and including this node have been applied to it.
\item $response$ is the result of the operation represented by the node.
\item $next$ is an instance of RC that will be used to agree upon the next node in the list.
\end{itemize}
Initially, the list contains a single dummy node whose sequence number is 1
and whose $newState$ field stores the initial state of the implemented object.

\ignore{
\begin{verbatim}
Initially, Announce[p] points to a dummy node which has sequence number $1$

Procedure ApplyOperation(opr: Pointer to opr record) {		// code for process i
       while (opr->seq = 0) {
               priority := (Head[i]->seq+1) mod n;
               if (Announce[priority]->seq = 0) then point := Announce[priority];
               else point := opr;
               win := decide(Head[i]->after, point);	
               win->new-state, win->response := apply(win->inv, Head[i]->new-state);
               win->seq = Head[i]->seq +1;
               Head[i] := win;	
      }
      return opr->response;
}

Universal(Invocation inv) {		// Code for process i: operation invocation, including parameters
       allocate a new opr record pointed to by Announce[i] with Announce[i]->inv := inv and Announce[i]->seq =0;
       for j := 0 to n-1 do	
            if (Head[j]->seq > Head[i]->seq) then Head[i] := Head[j];
       ApplyOperation(Announce[i]);
       return Announce[i]->response;
}

Procedure Recover() {   // code for process i
	ApplyOperation(Anounce[i])
}
\end{verbatim}
}

\newcommand{\com}{\>\>// }
\newcommand{\memb}{-\!\!>}
\begin{figure*}
{\small
\begin{code}
\firstline
shared variables:\nl
\n $Announce[1..n]$ of registers, each entry initially points to the dummy node at the beginning of the list\nl
   $Head[1..n]$, each entry initially points to a dummy node at the beginning of the list\bl
\p\nl
{\sc ApplyOperation} // ensures that node $Announce[i]$ is added to the list and returns result of that node's operation \nl
\n  while $Announce[i]{\memb}seq = 0$ \com keep trying until my operation has been added to the list\nl
\n     $priority \leftarrow (Head[i]{\memb}seq + 1) \mbox{ mod } n$ \com id of process who has priority for next list position\nl
       if $Announce[priority]{\memb}seq = 0$ then \com check if process with id $priority$ needs help\nl
\n        $pointer \leftarrow Announce[priority]$ \com try to add operation of process with id $priority$ \nl
\p     else\nl
\n        $pointer \leftarrow Announce[i]$     	\com try to add my own operation \nl
\p     end if\nl
       $winner \leftarrow \mbox{\sc Decide}(Head[i]{\memb}next, pointer)$  \com propose $pointer$ to RC instance associated with $next$ pointer of node $Head[i]$\nl	
       // fill in the information fields of $winner$, the next node in the list.\nl
       $\langle winner{\memb}newState, winner{\memb}response\rangle \leftarrow \mbox{\sc Apply}(winner{\memb}op, Head[i]{\memb}newState)$\nl
       $winner{\memb}seq \leftarrow Head[i]{\memb}seq +1$\nl
       $Head[i] \leftarrow winner$ \com advance to next node	\nl
\p  end while\nl
    return $Announce[i]{\memb}response$\nl
\p 
end {\sc ApplyOperation}\bl
\nl
{\sc Universal}$(op)$  // perform $op$ on implemented object and return result\nl
\n   $nd \leftarrow $ pointer to new list node\nl
     $nd{\memb}op \leftarrow op$ \com $op$ includes name of operation to apply and its arguments\nl
     $nd{\memb}seq \leftarrow 0$\nl
     $Announce[i] \leftarrow nd$\nl
     for $j \leftarrow 0..n-1$ \com make sure $Head[i]$ pointer is not too out of date\nl
\n       if $Head[j]{\memb}seq > Head[i]{\memb}seq$ then\nl
\n           $Head[i]\leftarrow Head[j]$\nl
\p       end if\nl
\p   end for\nl
     return {\sc ApplyOperation}\nl
\p end {\sc Universal}\bl
\nl
{\sc Recover}\nl
\n   return {\sc ApplyOperation}\nl
\p end {\sc Recover}
\end{code}
}
\caption{Universal construction pseudocode for process $p_i$. \label{universal-code}}
\end{figure*}

We remark that a process that crashes and recovers might access the RC instance associated 
with the $next$ pointer of a node multiple times with different input values.
So, we should use the mechanism described in the introduction to mask this behaviour
and ensure that the process's inputs to the RC instance are identical.

\section{Proof of Theorem \ref{rcons-implement}}
\label{rcons-implement-proof}

\begin{proof}
Let $A$ be an $n$-process algorithm for RC that uses
atomic objects of type $T_2$ and registers.
Construct an algorithm $A'$ by replacing each object of type $T_2$ by
a persistent linearizable implementation from atomic objects of type $T_1$ and registers.

We first argue that $A'$ satisfies the recoverable wait-freedom property.
If a process continues to take steps without crashing,
it will eventually complete each operation it calls on a simulted object of type $T_2$,
since the implementation of $T_2$ is wait-free.
Thus, it will eventually produce an output, since $A$ satisfies recoverable wait-freedom.

It remains to show that $A'$ satisfies the agreement and validity properties
of RC.
Let $\alpha'$ be any execution of $A'$.
We construct a corresponding execution $\alpha$ of $A$ as follows.
Remove all internal steps of the implementation of $T_2$ (i.e., all steps of a process between
an invocation step on a $T_2$ object and its subsequent response or process crash, or to the end of the execution if there is no such response or crash).
Each simulated operation on an object of type $T_2$ in $\alpha'$ that is not linearized
must not have a response in $\alpha'$.  We also remove its invocation when forming $\alpha$.
For each remaining operation on a $T_2$ object that has a response in $\alpha'$,
we ``contract'' the operation so that its invocation and response occur
immediately after each other at the linearization point of the operation.
Finally, we consider operations on $T_2$ objects that are linearized but have no response in $\alpha'$, either because
the process executing the operation crashes or does not take enough steps to complete the simulated 
operation.
We move the invocation step to the linearization point and add a response step immediately afterwards.
If the linearization point is \emph{after} the crash that occurred while the operation was pending,
then we also shift this crash step immediately after the response step.

It is easy to check that the constructed execution $\alpha$ is a legal execution of $A$ with atomic objects.
In particular, the sequence of steps taken by any process is the same in $\alpha$ as it is in $\alpha'$
(except for the removal of some invocations of operations on $T_2$ objects that do not terminate, either because they occur immediately before a crash of the process or because the process ceases taking steps). 
Thus, it will satisfy agreement and validity.
The execution $\alpha$ contains the same output steps as $\alpha'$,
so $\alpha'$ also satisfies agreement and validity.
\end{proof}

\section{$rcons(stack) = 1$}
\label{stack-proof}

\newcommand{\push}{\mbox{\sc Push}}
\newcommand{\pop}{\mbox{\sc Pop}}

We use a valency argument to show that $rcons(stack)=1$, i.e., that two processes cannot solve RC
using stacks and registers.
To derive a contradiction, assume there is an algorithm $A$ for two processes
to solve RC using stacks and registers.
As in the proof of Theorem \ref{necessary},
we define valency with respect to a set ${\mathcal E}_A$ of executions of $A$ in which
$p_2$ never crashes and, in any prefix of an execution, the number of crashes by $p_1$
is less than or equal to the number of steps taken by $p_2$.
By the same argument as in Theorem \ref{necessary},
there is a critical execution $\gamma$, 
and the next step specified by the algorithm for every process $p_i$ is an operation $op_i$ on a single stack $O$.
Let $v_i$ be the valency of the execution that is obtained by allowing $p_i$ to perform its next step $op_i$ after $\gamma$.
As argued in Theorem \ref{necessary}, $v_1$ must be different from $v_2$.
The remainder of the proof is a case analysis, similar to Herlihy's proof that 
$cons(stack)=2$ \cite{Her91}.  See Figure \ref{stack}, in which the sequence of
elements on the stack are shown from bottom to top
in the order they are pushed, and $\alpha$ represents a (possibly empty) sequence of elements.

\begin{figure*}
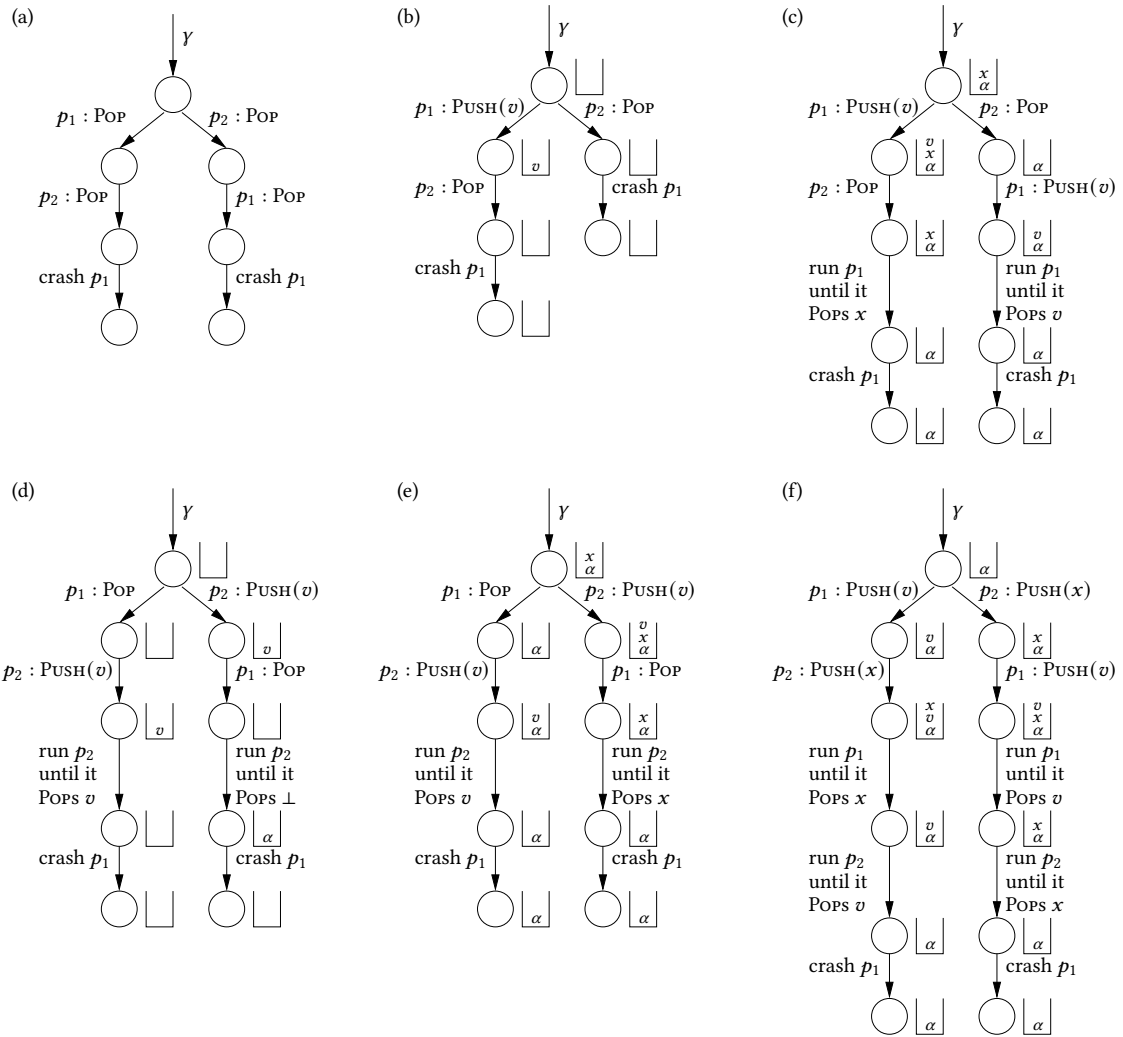
\caption{Impossibility of 2-process recoverable consensus using a stack. \label{stack}}
\end{figure*}

If both $op_1$ and $op_2$ are both \pop, then the steps commute (Figure \ref{stack}(a)).  
If $op_1$ is a \push$(v)$ and $op_2$ is a \pop, and the stack is empty at the end of $\gamma$,
then $op_1$ overwrites $op_2$ (Figure \ref{stack}(b)).  In either of these cases,
we have $v_1=v_2$ by Lemma \ref{same-valency} (which applies even when $O$ is not readable). 
This is the desired contradiction.

If $op_1$ is a \push$(v)$ and $op_2$ is a \pop\ and the stack is non-empty at the end
of $\gamma$, then consider the two extensions of $\gamma$ shown in Figure \ref{stack}(c).
After the operations $op_1$ and $op_2$ are done (in opposite order in the two extensions), 
the only differences between the two resulting states of the system are the local state of $p_2$
and the element that is on the top of the stack.  Thus, if $p_1$ continues to run, it must
run until it pops that top element; otherwise it would output the same value in both extensions,
contradicting the fact that one is $v_1$-valent and the other is $v_2$-valent.
During this solo execution by $p_1$ it takes the same steps in both extensions.
Since $p_2$ has taken a step in both extensions, we can then crash $p_1$.
After $p_1$ crashes, the states of the system in the two extensions are identical except for $p_2$'s state.
Thus, if $p_1$ recovers and executes $A$ to completion, it must output the same value in
the two extensions, contradicting the fact that the two extensions have different valencies.

The other cases shown in Figure \ref{stack}(d) to \ref{stack}(f) are argued similarly to 
the preceding case, and this completes the proof.

A similar argument could be used to show that $rcons(queue) =1$.

%% file: Tn-fig.pdf_t
\begin{picture}(0,0)%
\includegraphics{Tn-fig.pdf}%
\end{picture}%
\setlength{\unitlength}{3158sp}%
\begingroup\makeatletter\ifx\SetFigFont\undefined%
\gdef\SetFigFont#1#2#3#4#5{%
  \reset@font\fontsize{#1}{#2pt}%
  \fontfamily{#3}\fontseries{#4}\fontshape{#5}%
  \selectfont}%
\fi\endgroup%
\begin{picture}(5955,3736)(-14,-8075)
\put(1126,-5011){\makebox(0,0)[lb]{\smash{{\SetFigFont{8}{9.6}{\familydefault}{\mddefault}{\updefault}{\color[rgb]{0,0,0}$q_0$}%
}}}}
\put(2101,-4936){\makebox(0,0)[lb]{\smash{{\SetFigFont{8}{9.6}{\familydefault}{\mddefault}{\updefault}{\color[rgb]{0,0,0}$\tA$}%
}}}}
\put(3826,-4936){\makebox(0,0)[lb]{\smash{{\SetFigFont{8}{9.6}{\familydefault}{\mddefault}{\updefault}{\color[rgb]{0,0,0}$\tA$}%
}}}}
\put(3676,-5386){\makebox(0,0)[lb]{\smash{{\SetFigFont{8}{9.6}{\familydefault}{\mddefault}{\updefault}{\color[rgb]{0,0,0}$\tA$}%
}}}}
\put(4501,-4936){\makebox(0,0)[lb]{\smash{{\SetFigFont{8}{9.6}{\familydefault}{\mddefault}{\updefault}{\color[rgb]{0,0,0}$\tA$}%
}}}}
\put(3826,-5611){\makebox(0,0)[lb]{\smash{{\SetFigFont{8}{9.6}{\familydefault}{\mddefault}{\updefault}{\color[rgb]{0,0,0}$\tA$}%
}}}}
\put(4501,-5611){\makebox(0,0)[lb]{\smash{{\SetFigFont{8}{9.6}{\familydefault}{\mddefault}{\updefault}{\color[rgb]{0,0,0}$\tA$}%
}}}}
\put(4351,-5386){\makebox(0,0)[lb]{\smash{{\SetFigFont{8}{9.6}{\familydefault}{\mddefault}{\updefault}{\color[rgb]{0,0,0}$\tA$}%
}}}}
\put(5026,-5386){\makebox(0,0)[lb]{\smash{{\SetFigFont{8}{9.6}{\familydefault}{\mddefault}{\updefault}{\color[rgb]{0,0,0}$\tA$}%
}}}}
\put(5176,-5611){\makebox(0,0)[lb]{\smash{{\SetFigFont{8}{9.6}{\familydefault}{\mddefault}{\updefault}{\color[rgb]{0,0,0}$\tA$}%
}}}}
\put(5176,-4936){\makebox(0,0)[lb]{\smash{{\SetFigFont{8}{9.6}{\familydefault}{\mddefault}{\updefault}{\color[rgb]{0,0,0}$\tA$}%
}}}}
\put(5176,-6286){\makebox(0,0)[lb]{\smash{{\SetFigFont{8}{9.6}{\familydefault}{\mddefault}{\updefault}{\color[rgb]{0,0,0}$\tA$}%
}}}}
\put(5026,-6061){\makebox(0,0)[lb]{\smash{{\SetFigFont{8}{9.6}{\familydefault}{\mddefault}{\updefault}{\color[rgb]{0,0,0}$\tA$}%
}}}}
\put(4351,-6061){\makebox(0,0)[lb]{\smash{{\SetFigFont{8}{9.6}{\familydefault}{\mddefault}{\updefault}{\color[rgb]{0,0,0}$\tA$}%
}}}}
\put(3676,-6061){\makebox(0,0)[lb]{\smash{{\SetFigFont{8}{9.6}{\familydefault}{\mddefault}{\updefault}{\color[rgb]{0,0,0}$\tA$}%
}}}}
\put(3826,-6286){\makebox(0,0)[lb]{\smash{{\SetFigFont{8}{9.6}{\familydefault}{\mddefault}{\updefault}{\color[rgb]{0,0,0}$\tA$}%
}}}}
\put(4501,-6286){\makebox(0,0)[lb]{\smash{{\SetFigFont{8}{9.6}{\familydefault}{\mddefault}{\updefault}{\color[rgb]{0,0,0}$\tA$}%
}}}}
\put(5026,-6736){\makebox(0,0)[lb]{\smash{{\SetFigFont{8}{9.6}{\familydefault}{\mddefault}{\updefault}{\color[rgb]{0,0,0}$\tA$}%
}}}}
\put(4351,-6736){\makebox(0,0)[lb]{\smash{{\SetFigFont{8}{9.6}{\familydefault}{\mddefault}{\updefault}{\color[rgb]{0,0,0}$\tA$}%
}}}}
\put(3676,-6736){\makebox(0,0)[lb]{\smash{{\SetFigFont{8}{9.6}{\familydefault}{\mddefault}{\updefault}{\color[rgb]{0,0,0}$\tA$}%
}}}}
\put(1276,-5386){\makebox(0,0)[lb]{\smash{{\SetFigFont{8}{9.6}{\familydefault}{\mddefault}{\updefault}{\color[rgb]{0,0,0}$\tB$}%
}}}}
\put(1426,-5611){\makebox(0,0)[lb]{\smash{{\SetFigFont{8}{9.6}{\familydefault}{\mddefault}{\updefault}{\color[rgb]{0,0,0}$\tB$}%
}}}}
\put(2101,-5611){\makebox(0,0)[lb]{\smash{{\SetFigFont{8}{9.6}{\familydefault}{\mddefault}{\updefault}{\color[rgb]{0,0,0}$\tB$}%
}}}}
\put(2776,-5611){\makebox(0,0)[lb]{\smash{{\SetFigFont{8}{9.6}{\familydefault}{\mddefault}{\updefault}{\color[rgb]{0,0,0}$\tB$}%
}}}}
\put(2776,-6286){\makebox(0,0)[lb]{\smash{{\SetFigFont{8}{9.6}{\familydefault}{\mddefault}{\updefault}{\color[rgb]{0,0,0}$\tB$}%
}}}}
\put(2776,-6961){\makebox(0,0)[lb]{\smash{{\SetFigFont{8}{9.6}{\familydefault}{\mddefault}{\updefault}{\color[rgb]{0,0,0}$\tB$}%
}}}}
\put(2626,-7411){\makebox(0,0)[lb]{\smash{{\SetFigFont{8}{9.6}{\familydefault}{\mddefault}{\updefault}{\color[rgb]{0,0,0}$\tB$}%
}}}}
\put(2626,-6736){\makebox(0,0)[lb]{\smash{{\SetFigFont{8}{9.6}{\familydefault}{\mddefault}{\updefault}{\color[rgb]{0,0,0}$\tB$}%
}}}}
\put(2626,-6061){\makebox(0,0)[lb]{\smash{{\SetFigFont{8}{9.6}{\familydefault}{\mddefault}{\updefault}{\color[rgb]{0,0,0}$\tB$}%
}}}}
\put(1951,-6061){\makebox(0,0)[lb]{\smash{{\SetFigFont{8}{9.6}{\familydefault}{\mddefault}{\updefault}{\color[rgb]{0,0,0}$\tB$}%
}}}}
\put(1951,-6736){\makebox(0,0)[lb]{\smash{{\SetFigFont{8}{9.6}{\familydefault}{\mddefault}{\updefault}{\color[rgb]{0,0,0}$\tB$}%
}}}}
\put(1951,-7411){\makebox(0,0)[lb]{\smash{{\SetFigFont{8}{9.6}{\familydefault}{\mddefault}{\updefault}{\color[rgb]{0,0,0}$\tB$}%
}}}}
\put(1276,-7411){\makebox(0,0)[lb]{\smash{{\SetFigFont{8}{9.6}{\familydefault}{\mddefault}{\updefault}{\color[rgb]{0,0,0}$\tB$}%
}}}}
\put(1276,-6736){\makebox(0,0)[lb]{\smash{{\SetFigFont{8}{9.6}{\familydefault}{\mddefault}{\updefault}{\color[rgb]{0,0,0}$\tB$}%
}}}}
\put(1276,-6061){\makebox(0,0)[lb]{\smash{{\SetFigFont{8}{9.6}{\familydefault}{\mddefault}{\updefault}{\color[rgb]{0,0,0}$\tB$}%
}}}}
\put(1426,-6286){\makebox(0,0)[lb]{\smash{{\SetFigFont{8}{9.6}{\familydefault}{\mddefault}{\updefault}{\color[rgb]{0,0,0}$\tB$}%
}}}}
\put(2101,-6286){\makebox(0,0)[lb]{\smash{{\SetFigFont{8}{9.6}{\familydefault}{\mddefault}{\updefault}{\color[rgb]{0,0,0}$\tB$}%
}}}}
\put(2101,-6961){\makebox(0,0)[lb]{\smash{{\SetFigFont{8}{9.6}{\familydefault}{\mddefault}{\updefault}{\color[rgb]{0,0,0}$\tB$}%
}}}}
\put(1426,-6961){\makebox(0,0)[lb]{\smash{{\SetFigFont{8}{9.6}{\familydefault}{\mddefault}{\updefault}{\color[rgb]{0,0,0}$\tB$}%
}}}}
\put(1801,-8011){\makebox(0,0)[lb]{\smash{{\SetFigFont{8}{9.6}{\familydefault}{\mddefault}{\updefault}{\color[rgb]{0,0,0}$\floor{\frac{n}{2}}$ $op_\tA$'s}%
}}}}
\put(5926,-6061){\makebox(0,0)[lb]{\smash{{\SetFigFont{8}{9.6}{\familydefault}{\mddefault}{\updefault}{\color[rgb]{0,0,0}$\ceil{\frac{n}{2}}$ $op_\tB$'s}%
}}}}
\put(  1,-6061){\makebox(0,0)[lb]{\smash{{\SetFigFont{8}{9.6}{\familydefault}{\mddefault}{\updefault}{\color[rgb]{0,0,0}$\ceil{\frac{n}{2}}$ $op_\tB$'s}%
}}}}
\put(2551,-4486){\makebox(0,0)[lb]{\smash{{\SetFigFont{8}{9.6}{\familydefault}{\mddefault}{\updefault}{\color[rgb]{0,0,0}$\floor{\frac{n}{2}}$ $op_\tA$'s}%
}}}}
\end{picture}%

%% file: Sn-fig.pdf_t
\begin{picture}(0,0)%
\includegraphics{Sn-fig.pdf}%
\end{picture}%
\setlength{\unitlength}{3158sp}%
\begingroup\makeatletter\ifx\SetFigFont\undefined%
\gdef\SetFigFont#1#2#3#4#5{%
  \reset@font\fontsize{#1}{#2pt}%
  \fontfamily{#3}\fontseries{#4}\fontshape{#5}%
  \selectfont}%
\fi\endgroup%
\begin{picture}(2499,3399)(1939,-7948)
\put(3526,-6136){\makebox(0,0)[lb]{\smash{{\SetFigFont{8}{9.6}{\familydefault}{\mddefault}{\updefault}{\color[rgb]{0,0,0}$\tA,2$}%
}}}}
\put(2026,-7486){\makebox(0,0)[lb]{\smash{{\SetFigFont{8}{9.6}{\familydefault}{\mddefault}{\updefault}{\color[rgb]{0,0,0}$\tB,n-1$}%
}}}}
\put(3451,-7486){\makebox(0,0)[lb]{\smash{{\SetFigFont{8}{9.6}{\familydefault}{\mddefault}{\updefault}{\color[rgb]{0,0,0}$\tA,n-1$}%
}}}}
\put(2101,-4786){\makebox(0,0)[lb]{\smash{{\SetFigFont{8}{9.6}{\familydefault}{\mddefault}{\updefault}{\color[rgb]{0,0,0}$\tB,0$}%
}}}}
\put(2101,-5461){\makebox(0,0)[lb]{\smash{{\SetFigFont{8}{9.6}{\familydefault}{\mddefault}{\updefault}{\color[rgb]{0,0,0}$\tB, 1$}%
}}}}
\put(2101,-6136){\makebox(0,0)[lb]{\smash{{\SetFigFont{8}{9.6}{\familydefault}{\mddefault}{\updefault}{\color[rgb]{0,0,0}$\tB,2$}%
}}}}
\put(3526,-4786){\makebox(0,0)[lb]{\smash{{\SetFigFont{8}{9.6}{\familydefault}{\mddefault}{\updefault}{\color[rgb]{0,0,0}$\tA,0$}%
}}}}
\put(3526,-5461){\makebox(0,0)[lb]{\smash{{\SetFigFont{8}{9.6}{\familydefault}{\mddefault}{\updefault}{\color[rgb]{0,0,0}$\tA,1$}%
}}}}
\end{picture}%

%% file: stack-fig.pdf_t
\begin{picture}(0,0)%
\includegraphics{stack-fig.pdf}%
\end{picture}%
\setlength{\unitlength}{2960sp}%
\begingroup\makeatletter\ifx\SetFigFont\undefined%
\gdef\SetFigFont#1#2#3#4#5{%
  \reset@font\fontsize{#1}{#2pt}%
  \fontfamily{#3}\fontseries{#4}\fontshape{#5}%
  \selectfont}%
\fi\endgroup%
\begin{picture}(8891,8651)(811,-12000)
\put(8551,-11949){\makebox(0,0)[lb]{\smash{{\SetFigFont{6}{7.2}{\familydefault}{\mddefault}{\updefault}{\color[rgb]{0,0,0}$\alpha$}%
}}}}
\put(2326,-3661){\makebox(0,0)[lb]{\smash{{\SetFigFont{8}{9.6}{\familydefault}{\mddefault}{\updefault}{\color[rgb]{0,0,0}$\gamma$}%
}}}}
\put(2551,-4336){\makebox(0,0)[lb]{\smash{{\SetFigFont{8}{9.6}{\familydefault}{\mddefault}{\updefault}{\color[rgb]{0,0,0}$p_2:\pop$}%
}}}}
\put(2776,-5011){\makebox(0,0)[lb]{\smash{{\SetFigFont{8}{9.6}{\familydefault}{\mddefault}{\updefault}{\color[rgb]{0,0,0}$p_1:\pop$}%
}}}}
\put(2776,-5686){\makebox(0,0)[lb]{\smash{{\SetFigFont{8}{9.6}{\familydefault}{\mddefault}{\updefault}{\color[rgb]{0,0,0}crash $p_1$}%
}}}}
\put(1126,-5686){\makebox(0,0)[lb]{\smash{{\SetFigFont{8}{9.6}{\familydefault}{\mddefault}{\updefault}{\color[rgb]{0,0,0}crash $p_1$}%
}}}}
\put(1126,-5011){\makebox(0,0)[lb]{\smash{{\SetFigFont{8}{9.6}{\familydefault}{\mddefault}{\updefault}{\color[rgb]{0,0,0}$p_2:\pop$}%
}}}}
\put(1276,-4336){\makebox(0,0)[lb]{\smash{{\SetFigFont{8}{9.6}{\familydefault}{\mddefault}{\updefault}{\color[rgb]{0,0,0}$p_1:\pop$}%
}}}}
\put(8551,-4749){\makebox(0,0)[lb]{\smash{{\SetFigFont{6}{7.2}{\familydefault}{\mddefault}{\updefault}{\color[rgb]{0,0,0}$\alpha$}%
}}}}
\put(8551,-4644){\makebox(0,0)[lb]{\smash{{\SetFigFont{6}{7.2}{\familydefault}{\mddefault}{\updefault}{\color[rgb]{0,0,0}$x$}%
}}}}
\put(8551,-4538){\makebox(0,0)[lb]{\smash{{\SetFigFont{6}{7.2}{\familydefault}{\mddefault}{\updefault}{\color[rgb]{0,0,0}$v$}%
}}}}
\put(7576,-5611){\makebox(0,0)[lb]{\smash{{\SetFigFont{8}{9.6}{\familydefault}{\mddefault}{\updefault}{\color[rgb]{0,0,0}run $p_1$}%
}}}}
\put(7576,-5806){\makebox(0,0)[lb]{\smash{{\SetFigFont{8}{9.6}{\familydefault}{\mddefault}{\updefault}{\color[rgb]{0,0,0}until it}%
}}}}
\put(7576,-6001){\makebox(0,0)[lb]{\smash{{\SetFigFont{8}{9.6}{\familydefault}{\mddefault}{\updefault}{\color[rgb]{0,0,0}\pop s $x$}%
}}}}
\put(8776,-3586){\makebox(0,0)[lb]{\smash{{\SetFigFont{8}{9.6}{\familydefault}{\mddefault}{\updefault}{\color[rgb]{0,0,0}$\gamma$}%
}}}}
\put(7576,-4936){\makebox(0,0)[lb]{\smash{{\SetFigFont{8}{9.6}{\familydefault}{\mddefault}{\updefault}{\color[rgb]{0,0,0}$p_2:\pop$}%
}}}}
\put(7576,-4261){\makebox(0,0)[lb]{\smash{{\SetFigFont{8}{9.6}{\familydefault}{\mddefault}{\updefault}{\color[rgb]{0,0,0}$p_1:\push(v)$}%
}}}}
\put(9001,-4261){\makebox(0,0)[lb]{\smash{{\SetFigFont{8}{9.6}{\familydefault}{\mddefault}{\updefault}{\color[rgb]{0,0,0}$p_2:\pop$}%
}}}}
\put(9001,-4074){\makebox(0,0)[lb]{\smash{{\SetFigFont{6}{7.2}{\familydefault}{\mddefault}{\updefault}{\color[rgb]{0,0,0}$\alpha$}%
}}}}
\put(9001,-3969){\makebox(0,0)[lb]{\smash{{\SetFigFont{6}{7.2}{\familydefault}{\mddefault}{\updefault}{\color[rgb]{0,0,0}$x$}%
}}}}
\put(9451,-4749){\makebox(0,0)[lb]{\smash{{\SetFigFont{6}{7.2}{\familydefault}{\mddefault}{\updefault}{\color[rgb]{0,0,0}$\alpha$}%
}}}}
\put(9451,-5424){\makebox(0,0)[lb]{\smash{{\SetFigFont{6}{7.2}{\familydefault}{\mddefault}{\updefault}{\color[rgb]{0,0,0}$\alpha$}%
}}}}
\put(9451,-5319){\makebox(0,0)[lb]{\smash{{\SetFigFont{6}{7.2}{\familydefault}{\mddefault}{\updefault}{\color[rgb]{0,0,0}$v$}%
}}}}
\put(9226,-4936){\makebox(0,0)[lb]{\smash{{\SetFigFont{8}{9.6}{\familydefault}{\mddefault}{\updefault}{\color[rgb]{0,0,0}$p_1:\push(v)$}%
}}}}
\put(7576,-6511){\makebox(0,0)[lb]{\smash{{\SetFigFont{8}{9.6}{\familydefault}{\mddefault}{\updefault}{\color[rgb]{0,0,0}crash $p_1$}%
}}}}
\put(9226,-6511){\makebox(0,0)[lb]{\smash{{\SetFigFont{8}{9.6}{\familydefault}{\mddefault}{\updefault}{\color[rgb]{0,0,0}crash $p_1$}%
}}}}
\put(8551,-6999){\makebox(0,0)[lb]{\smash{{\SetFigFont{6}{7.2}{\familydefault}{\mddefault}{\updefault}{\color[rgb]{0,0,0}$\alpha$}%
}}}}
\put(8551,-6324){\makebox(0,0)[lb]{\smash{{\SetFigFont{6}{7.2}{\familydefault}{\mddefault}{\updefault}{\color[rgb]{0,0,0}$\alpha$}%
}}}}
\put(9451,-6324){\makebox(0,0)[lb]{\smash{{\SetFigFont{6}{7.2}{\familydefault}{\mddefault}{\updefault}{\color[rgb]{0,0,0}$\alpha$}%
}}}}
\put(9451,-6999){\makebox(0,0)[lb]{\smash{{\SetFigFont{6}{7.2}{\familydefault}{\mddefault}{\updefault}{\color[rgb]{0,0,0}$\alpha$}%
}}}}
\put(8551,-5424){\makebox(0,0)[lb]{\smash{{\SetFigFont{6}{7.2}{\familydefault}{\mddefault}{\updefault}{\color[rgb]{0,0,0}$\alpha$}%
}}}}
\put(8551,-5319){\makebox(0,0)[lb]{\smash{{\SetFigFont{6}{7.2}{\familydefault}{\mddefault}{\updefault}{\color[rgb]{0,0,0}$x$}%
}}}}
\put(9226,-5611){\makebox(0,0)[lb]{\smash{{\SetFigFont{8}{9.6}{\familydefault}{\mddefault}{\updefault}{\color[rgb]{0,0,0}run $p_1$}%
}}}}
\put(9226,-5806){\makebox(0,0)[lb]{\smash{{\SetFigFont{8}{9.6}{\familydefault}{\mddefault}{\updefault}{\color[rgb]{0,0,0}until it}%
}}}}
\put(9226,-6001){\makebox(0,0)[lb]{\smash{{\SetFigFont{8}{9.6}{\familydefault}{\mddefault}{\updefault}{\color[rgb]{0,0,0}\pop s $v$}%
}}}}
\put(4276,-9661){\makebox(0,0)[lb]{\smash{{\SetFigFont{8}{9.6}{\familydefault}{\mddefault}{\updefault}{\color[rgb]{0,0,0}run $p_2$}%
}}}}
\put(4276,-9856){\makebox(0,0)[lb]{\smash{{\SetFigFont{8}{9.6}{\familydefault}{\mddefault}{\updefault}{\color[rgb]{0,0,0}until it}%
}}}}
\put(4276,-10051){\makebox(0,0)[lb]{\smash{{\SetFigFont{8}{9.6}{\familydefault}{\mddefault}{\updefault}{\color[rgb]{0,0,0}\pop s $v$}%
}}}}
\put(6151,-8799){\makebox(0,0)[lb]{\smash{{\SetFigFont{6}{7.2}{\familydefault}{\mddefault}{\updefault}{\color[rgb]{0,0,0}$\alpha$}%
}}}}
\put(6151,-8694){\makebox(0,0)[lb]{\smash{{\SetFigFont{6}{7.2}{\familydefault}{\mddefault}{\updefault}{\color[rgb]{0,0,0}$x$}%
}}}}
\put(6151,-8588){\makebox(0,0)[lb]{\smash{{\SetFigFont{6}{7.2}{\familydefault}{\mddefault}{\updefault}{\color[rgb]{0,0,0}$v$}%
}}}}
\put(5476,-7636){\makebox(0,0)[lb]{\smash{{\SetFigFont{8}{9.6}{\familydefault}{\mddefault}{\updefault}{\color[rgb]{0,0,0}$\gamma$}%
}}}}
\put(5701,-8311){\makebox(0,0)[lb]{\smash{{\SetFigFont{8}{9.6}{\familydefault}{\mddefault}{\updefault}{\color[rgb]{0,0,0}$p_2:\push(v)$}%
}}}}
\put(5701,-8124){\makebox(0,0)[lb]{\smash{{\SetFigFont{6}{7.2}{\familydefault}{\mddefault}{\updefault}{\color[rgb]{0,0,0}$\alpha$}%
}}}}
\put(5701,-8019){\makebox(0,0)[lb]{\smash{{\SetFigFont{6}{7.2}{\familydefault}{\mddefault}{\updefault}{\color[rgb]{0,0,0}$x$}%
}}}}
\put(6151,-9474){\makebox(0,0)[lb]{\smash{{\SetFigFont{6}{7.2}{\familydefault}{\mddefault}{\updefault}{\color[rgb]{0,0,0}$\alpha$}%
}}}}
\put(6151,-9369){\makebox(0,0)[lb]{\smash{{\SetFigFont{6}{7.2}{\familydefault}{\mddefault}{\updefault}{\color[rgb]{0,0,0}$x$}%
}}}}
\put(5926,-8986){\makebox(0,0)[lb]{\smash{{\SetFigFont{8}{9.6}{\familydefault}{\mddefault}{\updefault}{\color[rgb]{0,0,0}$p_1:\pop$}%
}}}}
\put(5926,-10561){\makebox(0,0)[lb]{\smash{{\SetFigFont{8}{9.6}{\familydefault}{\mddefault}{\updefault}{\color[rgb]{0,0,0}crash $p_1$}%
}}}}
\put(5251,-11049){\makebox(0,0)[lb]{\smash{{\SetFigFont{6}{7.2}{\familydefault}{\mddefault}{\updefault}{\color[rgb]{0,0,0}$\alpha$}%
}}}}
\put(5251,-10374){\makebox(0,0)[lb]{\smash{{\SetFigFont{6}{7.2}{\familydefault}{\mddefault}{\updefault}{\color[rgb]{0,0,0}$\alpha$}%
}}}}
\put(6151,-10374){\makebox(0,0)[lb]{\smash{{\SetFigFont{6}{7.2}{\familydefault}{\mddefault}{\updefault}{\color[rgb]{0,0,0}$\alpha$}%
}}}}
\put(6151,-11049){\makebox(0,0)[lb]{\smash{{\SetFigFont{6}{7.2}{\familydefault}{\mddefault}{\updefault}{\color[rgb]{0,0,0}$\alpha$}%
}}}}
\put(5251,-9474){\makebox(0,0)[lb]{\smash{{\SetFigFont{6}{7.2}{\familydefault}{\mddefault}{\updefault}{\color[rgb]{0,0,0}$\alpha$}%
}}}}
\put(5251,-9369){\makebox(0,0)[lb]{\smash{{\SetFigFont{6}{7.2}{\familydefault}{\mddefault}{\updefault}{\color[rgb]{0,0,0}$v$}%
}}}}
\put(5926,-9661){\makebox(0,0)[lb]{\smash{{\SetFigFont{8}{9.6}{\familydefault}{\mddefault}{\updefault}{\color[rgb]{0,0,0}run $p_2$}%
}}}}
\put(5926,-9856){\makebox(0,0)[lb]{\smash{{\SetFigFont{8}{9.6}{\familydefault}{\mddefault}{\updefault}{\color[rgb]{0,0,0}until it}%
}}}}
\put(5926,-10051){\makebox(0,0)[lb]{\smash{{\SetFigFont{8}{9.6}{\familydefault}{\mddefault}{\updefault}{\color[rgb]{0,0,0}\pop s $x$}%
}}}}
\put(4501,-8311){\makebox(0,0)[lb]{\smash{{\SetFigFont{8}{9.6}{\familydefault}{\mddefault}{\updefault}{\color[rgb]{0,0,0}$p_1:\pop$}%
}}}}
\put(5251,-8799){\makebox(0,0)[lb]{\smash{{\SetFigFont{6}{7.2}{\familydefault}{\mddefault}{\updefault}{\color[rgb]{0,0,0}$\alpha$}%
}}}}
\put(9451,-8799){\makebox(0,0)[lb]{\smash{{\SetFigFont{6}{7.2}{\familydefault}{\mddefault}{\updefault}{\color[rgb]{0,0,0}$\alpha$}%
}}}}
\put(9451,-8694){\makebox(0,0)[lb]{\smash{{\SetFigFont{6}{7.2}{\familydefault}{\mddefault}{\updefault}{\color[rgb]{0,0,0}$x$}%
}}}}
\put(8551,-9474){\makebox(0,0)[lb]{\smash{{\SetFigFont{6}{7.2}{\familydefault}{\mddefault}{\updefault}{\color[rgb]{0,0,0}$\alpha$}%
}}}}
\put(8551,-9369){\makebox(0,0)[lb]{\smash{{\SetFigFont{6}{7.2}{\familydefault}{\mddefault}{\updefault}{\color[rgb]{0,0,0}$v$}%
}}}}
\put(8551,-9263){\makebox(0,0)[lb]{\smash{{\SetFigFont{6}{7.2}{\familydefault}{\mddefault}{\updefault}{\color[rgb]{0,0,0}$x$}%
}}}}
\put(9451,-9474){\makebox(0,0)[lb]{\smash{{\SetFigFont{6}{7.2}{\familydefault}{\mddefault}{\updefault}{\color[rgb]{0,0,0}$\alpha$}%
}}}}
\put(9451,-9369){\makebox(0,0)[lb]{\smash{{\SetFigFont{6}{7.2}{\familydefault}{\mddefault}{\updefault}{\color[rgb]{0,0,0}$x$}%
}}}}
\put(9451,-9263){\makebox(0,0)[lb]{\smash{{\SetFigFont{6}{7.2}{\familydefault}{\mddefault}{\updefault}{\color[rgb]{0,0,0}$v$}%
}}}}
\put(9451,-10374){\makebox(0,0)[lb]{\smash{{\SetFigFont{6}{7.2}{\familydefault}{\mddefault}{\updefault}{\color[rgb]{0,0,0}$\alpha$}%
}}}}
\put(9451,-10269){\makebox(0,0)[lb]{\smash{{\SetFigFont{6}{7.2}{\familydefault}{\mddefault}{\updefault}{\color[rgb]{0,0,0}$x$}%
}}}}
\put(8551,-10374){\makebox(0,0)[lb]{\smash{{\SetFigFont{6}{7.2}{\familydefault}{\mddefault}{\updefault}{\color[rgb]{0,0,0}$\alpha$}%
}}}}
\put(8551,-10269){\makebox(0,0)[lb]{\smash{{\SetFigFont{6}{7.2}{\familydefault}{\mddefault}{\updefault}{\color[rgb]{0,0,0}$v$}%
}}}}
\put(901,-3511){\makebox(0,0)[lb]{\smash{{\SetFigFont{8}{9.6}{\familydefault}{\mddefault}{\updefault}{\color[rgb]{0,0,0}(a)}%
}}}}
\put(4126,-3511){\makebox(0,0)[lb]{\smash{{\SetFigFont{8}{9.6}{\familydefault}{\mddefault}{\updefault}{\color[rgb]{0,0,0}(b)}%
}}}}
\put(7351,-3511){\makebox(0,0)[lb]{\smash{{\SetFigFont{8}{9.6}{\familydefault}{\mddefault}{\updefault}{\color[rgb]{0,0,0}(c)}%
}}}}
\put(5476,-3586){\makebox(0,0)[lb]{\smash{{\SetFigFont{8}{9.6}{\familydefault}{\mddefault}{\updefault}{\color[rgb]{0,0,0}$\gamma$}%
}}}}
\put(5701,-4261){\makebox(0,0)[lb]{\smash{{\SetFigFont{8}{9.6}{\familydefault}{\mddefault}{\updefault}{\color[rgb]{0,0,0}$p_2:\pop$}%
}}}}
\put(4276,-5611){\makebox(0,0)[lb]{\smash{{\SetFigFont{8}{9.6}{\familydefault}{\mddefault}{\updefault}{\color[rgb]{0,0,0}crash $p_1$}%
}}}}
\put(4276,-4936){\makebox(0,0)[lb]{\smash{{\SetFigFont{8}{9.6}{\familydefault}{\mddefault}{\updefault}{\color[rgb]{0,0,0}$p_2:\pop$}%
}}}}
\put(4276,-4261){\makebox(0,0)[lb]{\smash{{\SetFigFont{8}{9.6}{\familydefault}{\mddefault}{\updefault}{\color[rgb]{0,0,0}$p_1:\push(v)$}%
}}}}
\put(5926,-4936){\makebox(0,0)[lb]{\smash{{\SetFigFont{8}{9.6}{\familydefault}{\mddefault}{\updefault}{\color[rgb]{0,0,0}crash $p_1$}%
}}}}
\put(5251,-4749){\makebox(0,0)[lb]{\smash{{\SetFigFont{6}{7.2}{\familydefault}{\mddefault}{\updefault}{\color[rgb]{0,0,0}$v$}%
}}}}
\put(901,-7486){\makebox(0,0)[lb]{\smash{{\SetFigFont{8}{9.6}{\familydefault}{\mddefault}{\updefault}{\color[rgb]{0,0,0}(d)}%
}}}}
\put(4126,-7486){\makebox(0,0)[lb]{\smash{{\SetFigFont{8}{9.6}{\familydefault}{\mddefault}{\updefault}{\color[rgb]{0,0,0}(e)}%
}}}}
\put(7351,-7486){\makebox(0,0)[lb]{\smash{{\SetFigFont{8}{9.6}{\familydefault}{\mddefault}{\updefault}{\color[rgb]{0,0,0}(f)}%
}}}}
\put(2326,-7636){\makebox(0,0)[lb]{\smash{{\SetFigFont{8}{9.6}{\familydefault}{\mddefault}{\updefault}{\color[rgb]{0,0,0}$\gamma$}%
}}}}
\put(2551,-8311){\makebox(0,0)[lb]{\smash{{\SetFigFont{8}{9.6}{\familydefault}{\mddefault}{\updefault}{\color[rgb]{0,0,0}$p_2:\push(v)$}%
}}}}
\put(3001,-8799){\makebox(0,0)[lb]{\smash{{\SetFigFont{6}{7.2}{\familydefault}{\mddefault}{\updefault}{\color[rgb]{0,0,0}$v$}%
}}}}
\put(2776,-8986){\makebox(0,0)[lb]{\smash{{\SetFigFont{8}{9.6}{\familydefault}{\mddefault}{\updefault}{\color[rgb]{0,0,0}$p_1:\pop$}%
}}}}
\put(1126,-10561){\makebox(0,0)[lb]{\smash{{\SetFigFont{8}{9.6}{\familydefault}{\mddefault}{\updefault}{\color[rgb]{0,0,0}crash $p_1$}%
}}}}
\put(2776,-10561){\makebox(0,0)[lb]{\smash{{\SetFigFont{8}{9.6}{\familydefault}{\mddefault}{\updefault}{\color[rgb]{0,0,0}crash $p_1$}%
}}}}
\put(3001,-10374){\makebox(0,0)[lb]{\smash{{\SetFigFont{6}{7.2}{\familydefault}{\mddefault}{\updefault}{\color[rgb]{0,0,0}$\alpha$}%
}}}}
\put(2776,-9856){\makebox(0,0)[lb]{\smash{{\SetFigFont{8}{9.6}{\familydefault}{\mddefault}{\updefault}{\color[rgb]{0,0,0}until it}%
}}}}
\put(2101,-9474){\makebox(0,0)[lb]{\smash{{\SetFigFont{6}{7.2}{\familydefault}{\mddefault}{\updefault}{\color[rgb]{0,0,0}$v$}%
}}}}
\put(2776,-9661){\makebox(0,0)[lb]{\smash{{\SetFigFont{8}{9.6}{\familydefault}{\mddefault}{\updefault}{\color[rgb]{0,0,0}run $p_2$}%
}}}}
\put(1351,-8311){\makebox(0,0)[lb]{\smash{{\SetFigFont{8}{9.6}{\familydefault}{\mddefault}{\updefault}{\color[rgb]{0,0,0}$p_1:\pop$}%
}}}}
\put(2776,-10051){\makebox(0,0)[lb]{\smash{{\SetFigFont{8}{9.6}{\familydefault}{\mddefault}{\updefault}{\color[rgb]{0,0,0}\pop s $\bot$}%
}}}}
\put(826,-8986){\makebox(0,0)[lb]{\smash{{\SetFigFont{8}{9.6}{\familydefault}{\mddefault}{\updefault}{\color[rgb]{0,0,0}$p_2:\push(v)$}%
}}}}
\put(1126,-9661){\makebox(0,0)[lb]{\smash{{\SetFigFont{8}{9.6}{\familydefault}{\mddefault}{\updefault}{\color[rgb]{0,0,0}run $p_2$}%
}}}}
\put(1126,-9856){\makebox(0,0)[lb]{\smash{{\SetFigFont{8}{9.6}{\familydefault}{\mddefault}{\updefault}{\color[rgb]{0,0,0}until it}%
}}}}
\put(1126,-10051){\makebox(0,0)[lb]{\smash{{\SetFigFont{8}{9.6}{\familydefault}{\mddefault}{\updefault}{\color[rgb]{0,0,0}\pop s $v$}%
}}}}
\put(4276,-10561){\makebox(0,0)[lb]{\smash{{\SetFigFont{8}{9.6}{\familydefault}{\mddefault}{\updefault}{\color[rgb]{0,0,0}crash $p_1$}%
}}}}
\put(3976,-8986){\makebox(0,0)[lb]{\smash{{\SetFigFont{8}{9.6}{\familydefault}{\mddefault}{\updefault}{\color[rgb]{0,0,0}$p_2:\push(v)$}%
}}}}
\put(8776,-7636){\makebox(0,0)[lb]{\smash{{\SetFigFont{8}{9.6}{\familydefault}{\mddefault}{\updefault}{\color[rgb]{0,0,0}$\gamma$}%
}}}}
\put(9001,-8311){\makebox(0,0)[lb]{\smash{{\SetFigFont{8}{9.6}{\familydefault}{\mddefault}{\updefault}{\color[rgb]{0,0,0}$p_2:\push(x)$}%
}}}}
\put(9001,-8124){\makebox(0,0)[lb]{\smash{{\SetFigFont{6}{7.2}{\familydefault}{\mddefault}{\updefault}{\color[rgb]{0,0,0}$\alpha$}%
}}}}
\put(9226,-8986){\makebox(0,0)[lb]{\smash{{\SetFigFont{8}{9.6}{\familydefault}{\mddefault}{\updefault}{\color[rgb]{0,0,0}$p_1:\push(v)$}%
}}}}
\put(9226,-9661){\makebox(0,0)[lb]{\smash{{\SetFigFont{8}{9.6}{\familydefault}{\mddefault}{\updefault}{\color[rgb]{0,0,0}run $p_1$}%
}}}}
\put(9226,-9856){\makebox(0,0)[lb]{\smash{{\SetFigFont{8}{9.6}{\familydefault}{\mddefault}{\updefault}{\color[rgb]{0,0,0}until it}%
}}}}
\put(9226,-10051){\makebox(0,0)[lb]{\smash{{\SetFigFont{8}{9.6}{\familydefault}{\mddefault}{\updefault}{\color[rgb]{0,0,0}\pop s $v$}%
}}}}
\put(7576,-8311){\makebox(0,0)[lb]{\smash{{\SetFigFont{8}{9.6}{\familydefault}{\mddefault}{\updefault}{\color[rgb]{0,0,0}$p_1:\push(v)$}%
}}}}
\put(8551,-8799){\makebox(0,0)[lb]{\smash{{\SetFigFont{6}{7.2}{\familydefault}{\mddefault}{\updefault}{\color[rgb]{0,0,0}$\alpha$}%
}}}}
\put(8551,-8694){\makebox(0,0)[lb]{\smash{{\SetFigFont{6}{7.2}{\familydefault}{\mddefault}{\updefault}{\color[rgb]{0,0,0}$v$}%
}}}}
\put(7276,-8986){\makebox(0,0)[lb]{\smash{{\SetFigFont{8}{9.6}{\familydefault}{\mddefault}{\updefault}{\color[rgb]{0,0,0}$p_2:\push(x)$}%
}}}}
\put(7576,-9661){\makebox(0,0)[lb]{\smash{{\SetFigFont{8}{9.6}{\familydefault}{\mddefault}{\updefault}{\color[rgb]{0,0,0}run $p_1$}%
}}}}
\put(7576,-9856){\makebox(0,0)[lb]{\smash{{\SetFigFont{8}{9.6}{\familydefault}{\mddefault}{\updefault}{\color[rgb]{0,0,0}until it}%
}}}}
\put(7576,-10051){\makebox(0,0)[lb]{\smash{{\SetFigFont{8}{9.6}{\familydefault}{\mddefault}{\updefault}{\color[rgb]{0,0,0}\pop s $x$}%
}}}}
\put(7576,-11461){\makebox(0,0)[lb]{\smash{{\SetFigFont{8}{9.6}{\familydefault}{\mddefault}{\updefault}{\color[rgb]{0,0,0}crash $p_1$}%
}}}}
\put(9226,-11461){\makebox(0,0)[lb]{\smash{{\SetFigFont{8}{9.6}{\familydefault}{\mddefault}{\updefault}{\color[rgb]{0,0,0}crash $p_1$}%
}}}}
\put(7576,-10561){\makebox(0,0)[lb]{\smash{{\SetFigFont{8}{9.6}{\familydefault}{\mddefault}{\updefault}{\color[rgb]{0,0,0}run $p_2$}%
}}}}
\put(7576,-10756){\makebox(0,0)[lb]{\smash{{\SetFigFont{8}{9.6}{\familydefault}{\mddefault}{\updefault}{\color[rgb]{0,0,0}until it}%
}}}}
\put(7576,-10951){\makebox(0,0)[lb]{\smash{{\SetFigFont{8}{9.6}{\familydefault}{\mddefault}{\updefault}{\color[rgb]{0,0,0}\pop s $v$}%
}}}}
\put(9226,-10561){\makebox(0,0)[lb]{\smash{{\SetFigFont{8}{9.6}{\familydefault}{\mddefault}{\updefault}{\color[rgb]{0,0,0}run $p_2$}%
}}}}
\put(9226,-10756){\makebox(0,0)[lb]{\smash{{\SetFigFont{8}{9.6}{\familydefault}{\mddefault}{\updefault}{\color[rgb]{0,0,0}until it}%
}}}}
\put(9226,-10951){\makebox(0,0)[lb]{\smash{{\SetFigFont{8}{9.6}{\familydefault}{\mddefault}{\updefault}{\color[rgb]{0,0,0}\pop s $x$}%
}}}}
\put(8551,-11274){\makebox(0,0)[lb]{\smash{{\SetFigFont{6}{7.2}{\familydefault}{\mddefault}{\updefault}{\color[rgb]{0,0,0}$\alpha$}%
}}}}
\put(9451,-11274){\makebox(0,0)[lb]{\smash{{\SetFigFont{6}{7.2}{\familydefault}{\mddefault}{\updefault}{\color[rgb]{0,0,0}$\alpha$}%
}}}}
\put(9451,-11949){\makebox(0,0)[lb]{\smash{{\SetFigFont{6}{7.2}{\familydefault}{\mddefault}{\updefault}{\color[rgb]{0,0,0}$\alpha$}%
}}}}
\end{picture}%